 \documentclass[10pt,twocolumn]{IEEEtran}
\IEEEoverridecommandlockouts
\usepackage{cite}
\usepackage{amsmath, amssymb, amsfonts, mathtools}
\usepackage{bm}

\usepackage{subfig}

\usepackage[linesnumbered,ruled,vlined]{algorithm2e}
\usepackage{algpseudocode}
\usepackage{graphicx}
\usepackage{textcomp}
\usepackage{float}
\usepackage{xcolor}
\usepackage{array}
\usepackage{textcomp}
\usepackage{stfloats}
\usepackage{url}
\usepackage{verbatim}
\usepackage{flushend}
\usepackage{siunitx}
\usepackage{subcaption}
\usepackage[colorlinks=true, linkcolor=red]{hyperref}
\usepackage[letterpaper, right=0.62in, left=0.62in, top=0.7in, bottom=1.1in]{geometry}
\usepackage{booktabs}
\usepackage{enumitem}
\setlist[itemize]{leftmargin=*}

\newtheorem{assumption}{Assumption}
\newtheorem{definition}{Definition}
\newtheorem{lemma}{Lemma}
\newtheorem{theorem}{Theorem}
\newtheorem{remark}{Remark}

\definecolor{blue}{named}{black}

\newcommand{\bOmega}{\boldsymbol{\Omega}}

\newcommand{\cD}{{\cal D}}
\newcommand{\cE}{{\cal E}}

\newcommand{\cM}{{\cal M}}
\newcommand{\cO}{{\cal O}}
\newcommand{\cR}{{\cal R}}

\newcommand{\mE}{\mathbb{E}}

\newcommand{\btomega}{\boldsymbol{\tilde \omega}}

\def\BibTeX{{\rm B\kern-.05em{\sc i\kern-.025em b}\kern-.08em
    T\kern-.1667em\lower.7ex\hbox{E}\kern-.125emX}}
\begin{document}
\title{Skillsets on the Chain: A Blockchain-based Zero-Trust Framework for Agentic AI Networking}

\author{Yayu~Gao, \IEEEmembership{Member, IEEE}, Yong~Xiao, \IEEEmembership{Senior~Member, IEEE}, Hao Hu, Xubo Li, Zhiwei Liu, Yingyu~Li, Guangming~Shi, \IEEEmembership{Fellow, IEEE}, and Ping Zhang, \IEEEmembership{Fellow, IEEE}
\thanks{*This work is accepted at IEEE Transactions on Cognitive Communications and Networking. Copyright may be transferred without notice, after which this version may no longer be accessible.

This work was supported in part by the National Natural Science Foundation of China (NSFC) under grants 62571208 and 62525109, the Mobile Information Network National Science and Technology Key Project under grant 2024ZD1300700, and Hubei Natural Science Foundation Innovation Research Group Program under grant 2026AFA044. An earlier version of this paper was presented in part at the Proceedings of the IEEE GLOBECOM, Taipei, Taiwan, December 2025\cite{GaoYY_TrustAgentNet}. (Corresponding author: Yong Xiao.).

Yayu~Gao, Yong~Xiao, Xubo~Li and Zhiwei Liu are with the School of Electronic Information and Communications, the Huazhong University of Science and Technology, Wuhan, China 430074. Yong~Xiao is also with the Peng Cheng Laboratory, Shenzhen, China, and Pazhou Laboratory (Huangpu), Guangzhou, China (e-mail: \{yayugao, yongxiao, xuboli, zhiweiliu\}@hust.edu.cn). 

Hao~Hu and Yingyu Li are with the School of Mechanical Engineering and Electronic Information, China University of Geosciences (Wuhan), Wuhan, China 430074 (e-mail: \{1202520838, liyingyu29\}@cug.edu.cn).

G. Shi is with the Peng Cheng Laboratory, Shenzhen, China 518055, also with the School of Artificial Intelligence, Xidian University, Xi'an, Shaanxi, China  710071 (e-mail: gmshi@xidian.edu.cn). 

P. Zhang is with the State Key Laboratory of Networking and Switching Technology, Beijing University of Posts and Telecommunications, Beijing, China 100876 (email: pzhang@bupt.edu.cn).
}
}
 \setlength{\abovedisplayskip}{2.8pt}
\maketitle

\vspace{-12mm}
\begin{abstract}
\textcolor{blue}{Agentic AI networking (AgentNet) systems rely heavily on third-party skillset implementations and distributed multi-agent collaboration, yet they face major claim-to-capability inconsistencies and security vulnerabilities under trust-by-declaration assumptions. To bridge this gap, this paper proposes \emph{TrustAgentNet}, a dual-tier blockchain-secured zero-trust framework. Specifically, a global Chain of Skillsets (CoS) governs the lifecycle of skillset metadata with protocols empowered by specialized agents to enforce off-chain auditing while maintaining lightweight on-chain cryptographic consensus. Furthermore, transient, task-oriented Chains of Collaboration (CoC) are dynamically established to enable trustless distributed multi-agent collaboration. Theoretical analysis of the three-way trade-off among security level, task performance, and resource overhead is provided and empirically validated. Experimental results on a hardware prototype demonstrate that compared with no-blockchain trust-by-default baselines, the zero-trust overhead of TrustAgentNet is dominated by off-chain inference, while the blockchain layer incurs minor ledger costs via the ledger-IPFS storage and on/off-chain integration design. Crucially, the proposed verification pipeline achieves a flawless $100\%$ accuracy across $50$ AI models, correctly validating $40$ honest skillsets and intercepting $10$ adversarial ones, and generalizes to non-AI domains with an $83.91\%$ accuracy and a $0.85$ F1-score across $1478$ features from $171$ ClawHub skills. Adversarial experiments further show that TrustAgentNet enables autonomous skillset self-recovery against various malicious attacks.}
\end{abstract}


\section{Introduction}
Agentic AI networking (AgentNet) is a novel AI-native networking ecosystem in which autonomous AI agents can collaborate, reason, and plan to solve complex, multi-step problems with minimal human intervention. It has the potential to overcome the limitations of existing AI-based solutions, positioning it as one of the possible architectures for next-generation networking systems, especially 6G and beyond\cite{xiao2025AgentNet}. \textcolor{blue}{The recent rapid rise of OpenClaw has empowered agents to independently invoke and implement third-party capabilities, opening up unprecedented possibilities for agentic autonomy and multi-agent collaborative workflows to complete complex tasks. However, the recent emergence of ClawdBot anomalies \cite{guardz2026clawdbot} exposes systemic vulnerabilities under this paradigm. Driven by sophisticated reasoning but operating on unverified third-party skillsets, agents can inadvertently or maliciously initiate unauthorized cross-platform interactions and recursive scanning.} 


\textcolor{blue}{A skillset is defined as a modular, structured capability unit, comprising either neural network assets (e.g., AI model weights and gradients) or executable programmatic logic (e.g., procedural scripts provided in ClawHub), that provides specific functions, operational rules, and external dependencies to the invoking agent. Currently, the burgeoning skillset supply chain heavily relies on a ``trust-by-declaration" model, where agents and skillsets all declare their functional capabilities using high-level natural-language directives in their metadata. These claimed capabilities are further encapsulated into standardized digital manifests, such as agent cards \cite{a2aspec2025}, and continuously broadcasted across AgentNet systems. When agents are required to solve complex distributed tasks, network orchestrators select candidate experts and route critical operation paths based solely on these unverified, advertised capability metadata.}

\textcolor{blue}{However, this structural reliance on unverified advertisements creates a critical governance vacuum and gives rise to a dangerous \emph{Claim-to-Capability} gap. Recent large-scale empirical studies reveal a pervasive and systemic security crisis, showing over 26\% of deployed skills harbor severe vulnerabilities \cite{liu2026agentskillswildempirical}, and 80\% of active skills functionally deviate from their declared specifications \cite{wu2026behavioralintegrityverificationai}. When untrusted skillsets are integrated into agents' decision logic, the claim-to-capability mismatch would propagate through the system layers, inducing severe performance drift, semantic deviation, and volatile execution boundaries across the agents' subsequent behavioral trajectories. Consequently, the lifecycle security of third-party capabilities inside highly autonomous AgentNet systems is deeply coupled with a critical trusted supply chain problem, which further propagates into a collaboration trust crisis during multi-agent interactive workflows.}

\textcolor{blue}{To mitigate these vulnerabilities, the evolution of AgentNet necessitates the adoption of a decentralized, zero-trust framework predicated on the principle of \emph{never trust, always verify} \cite{Sedjelmaci2024,Chen2024}. However, enforcing continuous, verifiable full-lifecycle governance for skillsets and agents across untrusted edge environments poses the following three pivotal challenges:} 

\noindent
{\bf (1) Zero-trust \textcolor{blue}{lifecycle management of diverse agents' skillsets}:} 
    A central challenge in the deployment of AgentNet systems is \textcolor{blue}{establishing a secure registry that governs the entire lifecycle of agent skillsets without relying on a centralized authority. While baseline cryptographic identity provisioning can authenticate the digital identity of an agent, it cannot dynamically track or govern its evolving skillsets. It necessitates the development of a unified framework that can securely and efficiently manage and store, dynamically update and audit the versioning and declarations of these diverse skillsets and agents under strict zero-trust constraints}.

\noindent
{\bf (2) Zero-trust \textcolor{blue}{claim-to-capability consistency auditing of skillsets}:} 
    To ensure efficient multi-agent planning and task allocation, the network infrastructure should guarantee that each agent's claimed capabilities match its real performance. \textcolor{blue}{To move beyond the trust-by-declaration model, the key lies in designing a verification mechanism for a skillset to identify the capability-performance gap before it is authorized for wide deployment. However, executing complex semantic evaluation and consistency verification algorithms following the zero-trust principle can be computationally prohibitive for resource-constrained blockchain ledgers, which demands further investigation.}
    
 \noindent
{\bf (3) Privacy preserving and secure sharing in task-oriented multi-agent collaboration:}  
    \textcolor{blue}{Orchestrating task-oriented multi-agent collaborative training and inference while preserving their privacy without sharing local raw data can be challenging. Moreover, these security enhancements inevitably introduce considerable computational latency and communication signaling overhead, leading to a tradeoff among security levels, task performance, and resource costs. Addressing this trade-off is essential to ensure that decentralized multi-agent collaboration remains scalable, flexible, and resilient under dynamic network environments.}
    
    
To address the above challenges, in this paper, we introduce a zero-trust AgentNet architecture, called TrustAgentNet, that supports skillset lifecycle management  and decentralized multi-agent collaboration. Our main contributions are summarized as follows:

\noindent{\textcolor{blue}{\bf (1) Dual-tier blockchain-enabled zero-trust architecture}}: \textcolor{blue}{We design a hierarchical, permissioned zero-trust framework. Specifically, a global Chain of Skillset (CoS) serves as the immutable unified source of truth for agent identities and verified skillsets,decoupling the lightweight on-chain metadata registry from off-chain large model storage to minimize ledger bloat while guaranteeing content integrity. Moreover, transient task-oriented Chains of Collaboration (CoC) are dynamically orchestrated upon task instantiation and dissolved upon completion to support decentralized, privacy-preserving multi-agent collaboration, ensuring scalable and flexible balance between systemic security and resource overhead.}

\noindent{\textcolor{blue}{\bf (2) LLM-agent-enabled skillset lifecycle management protocols:}} \textcolor{blue}{We propose a secure skillset acquisition and semantic mapping protocol and a skillset submission and consistency verification protocol by utilizing specialized LLM agents. By shifting semantic inference and cross-modal consistency validation off-chain, while anchoring cryptographic hashes and trust state updates as immutable blockchain transactions, the protocols successfully secure both skillset registration and retrieval while maintaining lightweight on-chain consensus.}
   
\noindent{\bf (3) Theoretical analysis of three-way tradeoff among security level, agent performance, and resource cost:} We provide a theoretical analysis that captures the three-way tradeoff among the security level, the agents' model performance, as well as the communication and computational resource costs of the proposed TrustAgentNet. This tradeoff is further empirically validated by the experimental results obtained from a hardware prototype implemented based on a Hyperledger Fabric-based consortium blockchain. Experimental results further
    validate that the level of security can be improved at the cost of traffic volume, computational demands and running time for skillset training under a given performance target. 
    
    \noindent{\bf (4) Prototype validation and performance analysis:} We develop a 5G-enabled hardware/software prototype integrating Hyperledger Fabric and Kubo IPFS to evaluate the performance of TrustAgentNet. \textcolor{blue}{Resource consumption comparison with no-blockchain baseline reveals that the zero-trust overhead is heavily dominated by off-chain inference pipelines, whereas the blockchain layer imposes minimal constant-time transaction commitment delays owing to the blockchain-IPFS synergy and on-chain logging/off-chain inference design. Evaluated against an adversarial dataset of $50$ AI models, the framework achieves a flawless $100\%$ inconsistency interception rate. Furthermore, the verification protocol yields an $83.91\%$ accuracy and a $0.85$ F1-score against $1478$ functional features from $171$ real-world ClawHub procedural skills, demonstrating its potential to govern diverse types of skillsets. Adversarial experiments demonstrate that agents can autonomously recover from diverse security risks via CoS interactions.}
    
    

The remainder of this paper is organized as follows. Section \ref{sec:relatedwork} summarizes the related work. Section \ref{sec:model} introduces the system model and problem statement. Section \ref{sec:framework} introduces the architectural framework of TrustAgentNet. \textcolor{blue}{Section \ref{sec:protocol} presents the LLM-driven skillset lifecycle management protocols, and Section \ref{sec:coc} introduces the CoC-oriented multi-agent collaboration methodology and its security-performance-resource trade-off analysis.} A prototype of TrustAgentNet is implemented and evaluated in Section \ref{sec:evaluation}. Concluding remarks are summarized in Section \ref{sec:conclusion}.

\section{Related Work}
\label{sec:relatedwork}
\subsection{Agentic AI and AgentNet}

Agentic AI systems have attracted significant interest recently due to their potential to fundamentally shift the paradigm from reactive, passive, and monolithic learning-based models to autonomous, goal-driven, multi-agent cooperation architectures based on decentralized learning and solution finding. Critically, these systems are evolving towards a decentralized structure where heterogeneous AI agents--each possessing specialized expertise and unique optimization targets--can collaborate to fulfill complex objectives directly on peer-to-peer networks. This decentralized approach inherently mitigates the risks associated with single points of failure and addresses the scalability and latency issues prevalent in centralized frameworks. Most existing studies primarily investigate the reasoning and orchestration capacities for multi-agent planning, tool use, memory integration, and environmental perception to enable complex, multi-step execution \cite{Durante2024AgentAISurvey}.

Building on these foundations, AgentNet was proposed as an AI-native networking paradigm characterized by high degrees of autonomy and adaptability, allowing agents to pursue objectives with minimal human intervention\cite{acharya2025agentic}. Our recent work \cite{xiao2025AgentNet} proposed a Generative Foundation Model (GFM)-based architecture to facilitate interaction, collaborative learning, and efficient knowledge synthesis among multiple GFM-as-agents. The practical utility of the AgentNet framework was demonstrated through 6G use cases, digital-twin industrial automation and metaverse infotainment, highlighting how AgentNet facilitates task-driven interactive networking systems with diverse AI agents. 

\subsection{\textcolor{blue}{Safety, Security, and Governance for Agentic AI}}

\textcolor{blue}{Due to the profound and novel security risks of agentic AI systems including runaway behavior and multi-agent cascades, AI agent governance has become vital yet remains highly challenging \cite{kraprayoon2025aiagentgovernancefield}. Several studies have begun integrating zero-trust architectures into agentic AI, including continuous semantic intent validation for LLM-driven O-RAN control \cite{wang2026promptguard}, zero-trust approach for MCP-based agents \cite{yoshi2026zerotrust} and initial protocols for agent identity verification and delegation chains \cite{Robert2026zerotrust}.}

Meanwhile, the security of LLM-based agents has become a critical research frontier, leading to the development of specialized benchmarks and systematic defense frameworks. \cite{shen2024ccs} demonstrated that even well-aligned LLMs and their safeguards remain significantly vulnerable to diverse attack strategies—such as virtualization, deception, and privilege escalation. AgentHarm \cite{andriushchenko2025agentharm} measured agentic robustness across $11$ harm categories through explicitly malicious tasks; and \cite{zhang2025asb} proposed a comprehensive framework that benchmarked a wide array of attacks and defenses across diverse real-world scenarios. To further enhance evaluation precision, ASSEBench \cite{luo2025agentauditor} provided a meticulously annotated dataset for testing the risk-detection capabilities of LLM-based safety evaluators. On the defensive side, holistic blueprints were proposed \cite{ying2026uncovering} to emphasize zero-trust execution, dynamic intent verification, and reasoning-action correlation.

Following the rise of high-autonomy frameworks like OpenClaw, which grant agents operating-system-level permissions, research focus has been increasingly shifting from prompt-centric safety toward the systemic security vulnerabilities inherent in agent skills and tool invocations. \cite{li2026secureagentskillsarchitecture} established a pivotal four-stage lifecycle security model and a seven-category threat taxonomy, identifying structural flaws inherent in the skill framework. Through large-scale empirical analysis, \cite{liu2026agentskillswildempirical,liu2026domentionuserdetecting,xie2025toolsafety} have uncovered a critical security crisis in the agentic ecosystem, revealing that approximately 26.1\% of skills harbor vulnerabilities like data exfiltration and privilege escalation, highlighting significant defensive bottlenecks in managing multi-step tool interactions and indirect harm scenarios. \cite{duan2026skillattackautomatedredteaming} further demonstrated that these latent flaws are actively exploitable through iterative, feedback-driven adversarial prompting without any modification to the skills themselves. Recently, \cite{wu2026behavioralintegrityverificationai} introduced a behavioral integrity verification (BIV) mechanism to address the declaration-implementation gap within open skill registries (e.g., OpenClaw) via code analysis, uncovering that 80\% of skills functionally deviate from their declared capabilities.

\subsection{Blockchain-enabled Distributed Learning and Computing}

Tracing its origins to Satoshi Nakamoto’s 2008 seminal work on Bitcoin, blockchain technology has evolved from a peer-to-peer cryptocurrency system into a cornerstone for ensuring trust in distributed ecosystems. At its core, a blockchain is a distributed, immutable ledger that utilizes cryptographic chaining to secure data blocks, thereby facilitating transparency and permanent record-keeping in the absence of a governing intermediary. Key features such as smart contracts, consensus algorithms, and robust encryption have driven its application far beyond its initial financial scope. Consequently, blockchain has improved operations in diverse domains such as logistics, decentralized finance, and digital identity, and is increasingly being integrated into distributed learning/computing environments to address data privacy and systemic integrity\cite{10528325,10.1145/3524104}. This paradigm shift is particularly evident in the Federated Learning (FL) frameworks. By replacing the traditional central aggregator with a distributed ledger, blockchain-enabled FL was found to effectively mitigate single-point-of-failure risks, model tampering, and poisoning attacks for collaborative training processes across medical and industrial domains \cite{When_Federated_Learning_Meets_Blockchain,10.1145/3501813}. \cite{9761745} further combined Byzantine Fault-Tolerant (BFT) aggregation with differential privacy, mathematically proving model convergence under malicious attacks. To address the computational overhead of decentralized ledgers, computation reuse mechanisms were proposed in \cite{8843900} to reduce the energy consumption of the blockchain-enabled FL training process. \textcolor{blue}{\cite{WANG2023127} leveraged blockchain to establish a secure, incentive-aware collaboration mechanism that motivates edge nodes to participate in federated training for joint caching and computation offloading optimization.} Recently, \cite{zhu2025moemeetsblockchaintrustworthy} proposed B-MoE, utilizing blockchain to trace, verify, and record computational results within distributed Mixture of Experts (MoE) in large models.

Beyond model training, blockchain also serves as a cornerstone for zero-trust security and resource management in B5G/6G edge intelligence environments \cite{8818339,M202248,XU2020261,8726067}. Specifically, optimized consensus algorithms like PBFT were deployed to achieve end-to-end traceability of user activities \cite{8818339}. The synergy between blockchain and intelligent systems is further explored in vehicular networks, where blockchain assisted efficient batch authentication and key exchange for security rating prediction \cite{M202248}, and in network slicing, where blockchain facilitated real-time resource monitoring to meet sub-1ms latency constraints \cite{XU2020261}.  \cite{8726067} proposed the integration of deep reinforcement learning with blockchain to optimize dynamic content caching in time-variant environments. 

\subsection{\textcolor{blue}{Remarks}}

\textcolor{blue}{In summary, while Agentic AI systems and AgentNet are still in the early stages of development, both academia and industry have increasingly recognized the unprecedented security and trust challenges they introduce. The continuous evolution is transiting from prompt-centric defense vectors toward systemic skillset safety. However, existing paradigms predominantly operate as static, centralized toolkits localized within individual environments, leaving a critical gap in enforcing decentralized, full-lifecycle skillset governance and runtime resilience. Meanwhile, although the integration of blockchain technology into distributed learning and edge computing provides a promising technical pathway for decentralized trust, existing blockchain paradigms cannot be directly applied to address the unique challenges of skillset lifelong governance, including task-to-skillset mapping and claim-to-capability consistency verification, which may lead to profound computing and storage consumption bottlenecks on the chain. Furthermore, the mathematical trade-off among security enhancements, task performance, and resource overhead introduced by anchoring ledgers into multi-agent systems remains largely unexplored. To bridge these critical gaps, this paper proposes the zero-trust TrustAgentNet framework, designs blockchain-secured skillset lifecycle management algorithms, and establishes secure decentralized on-chain multi-agent collaboration.}

\section{System Model and Problem Formulation}
\label{sec:model}
\subsection{System Model}
We consider a general AgentNet system architecture structured as a decentralized, non-perimeter network. The network consists of a set $\mathcal{A}$ of heterogeneous agents, deployed across various environments $\mathcal{E}$. In each environment $e$, $\mathcal{A}_e$ agents are deployed and can collaborate to solve a finite set of tasks $\mathcal{T}_e$. Under the zero-trust paradigm, no agent $a_k \in \mathcal{A}$ is inherently trusted based on its network location or origin. Instead, every agent is treated as a latent threat until its identity and operational integrity are verified.

Each task can be further decomposed into multiple sub-tasks, each requiring a specific skillset generating the intended output based on the input. \textcolor{blue}{To be specific, a \emph{skillset} is generally defined as a collection of encapsulated capabilities that dictates how an agent executes a specific action or solves a sub-task. Broadly speaking, the realization of a skillset can be diverse, ranging from deterministic scripts and external API calls to learnable AI models. In this paper, we specifically focus on AI-model-based skill realization, where the logic is embedded within a learnable parametric structure. Formally, we define a skillset $s_i \in \mathcal{S}$ as a functional entity characterized by the tuple:}
\begin{equation}
    \label{eq_skillset}
  \textcolor{blue}{  s_i \triangleq \langle D_i, \mathcal{X}_i, \mathcal{Y}_i, \Phi_i, \omega_i \rangle,}
\end{equation}
\textcolor{blue}{where $D_i$ represents the declarative semantic metadata (e.g., natural-language declarations) advertised by the skillset provider, $\mathcal{X}_i$ and $\mathcal{Y}_i$ represent the input space and output space respectively; $\Phi_i: \mathcal{X}_i \times \Omega_i \to \mathcal{Y}_i$ denotes the mapping functional logic (e.g., the model architecture or executable code), and $\omega_i \in \Omega_i$  denotes the learnable weights/parameters of the skillset.} Assume that each agent $a_k$ in environment $e$ can access a subset of skillsets $\mathcal{S}_{k,e} {=} \{s_1, \dots, s_i\}$. Each agent maintains an exclusive local dataset $\mathcal{D}_{k,e}$ sampled from an unknown distribution $P_{k,e}$. To satisfy zero-trust privacy requirements, raw data $\mathcal{D}_{k,e}$ is never exposed; only verifiable model updates or proofs of training can be transmitted.

\textcolor{blue}{To enforce the ``\emph{never trust, always verify}" principle, we propose a hierarchical, blockchain-anchored architecture, TrustAgentNet, consisting of the Chain of Skillsets (CoS) and the Chain of Collaboration (CoC). In particular, the CoS serves as the global, consortium-blockchain-based ledger that permanently anchors the lifecycle of the universal skillset set $\mathcal{S}$. To alleviate the on-chain storage bottleneck while preserving cryptographic immutability, only the semantic declaration $D_i$ and the cryptographic hashes of the skillsets are stored on the CoS ledger as verification baselines, whereas the corresponding full data are securely offloaded to a distributed storage infrastructure. Each agent $a_k$ with authorized identity can interact with the CoS to execute secure skillset acquisition and submission.}

\textcolor{blue}{Meanwhile, CoC is architected as an on-demand, task-specific blockchain embedded within CoS. When a task publisher schedules a real-time collaborative session, a transient CoC is spawned from the CoS to orchestrate the distributed coordination among a subset of authenticated agents. Each CoC acts as a distributed ledger for gradient variations, intermediate embeddings etc. during the collaboration. Upon task completion, the consolidated updates achieved on the CoC are immutably settled back into the global CoS ledger, and the resources of CoC are released. Under this framework, authorized distributed agents can securely collaborate across heterogeneous environments without compromising data privacy.}



\subsection{\textcolor{blue}{Problem Formulation and Multi-Tier Zero-Trust Trade-off Metrics}}

\textcolor{blue}{In traditional AgentNet architectures, skillsets are uploaded, downloaded, and executed across agents under an implicit trust-by-declaration assumption, introducing severe vulnerabilities. To address these limitations, \textit{TrustAgentNet} replaces implicit trust with a zero-trust framework. However, enforcing continuous verification may incur multi-dimensional resource and performance penalties.}

\textcolor{blue}{To formalize these trade-offs mathematically, we decouple the system orchestration into two logical tiers: \textbf{Tier-1: CoS-oriented skillset lifecycle management}, which governs the secure submission and acquisition of skillsets; and \textbf{Tier-2: CoC-oriented multi-agent collaboration}, which regulates on-demand task-oriented collaborative learning/inference among multiple agents.}

\subsubsection{\textcolor{blue}{Tier-1: CoS-Oriented Skillset Lifecycle Management (Two-Way Trade-off)}}
\textcolor{blue}{This tier governs the static lifecycle of skillsets interacting with CoS. Let us define the binary action variables $\boldsymbol{\pi}$ for a specific agent $a_j \in \mathcal{A}$ and a skillset $s_i \in \mathcal{S}$: the submission policy $\pi^{(j,i)}_{\text{sub}} \in \{0, 1\}$ (where $\pi^{(j,i)}_{\text{sub}}=1$ mandates rigorous on-chain verification when agent $a_j$ submits skillset $s_i$) and the acquisition policy $\pi^{(j,i)}_{\text{acq}} \in \{0, 1\}$ (where $\pi^{(j,i)}_{\text{acq}}=1$ mandates smart-contract-based authenticated retrieval when agent $a_j$ requests skillset $s_i$). The trade-off is inherently two-dimensional:}
\begin{itemize}
    \item \textcolor{blue}{\textbf{Lifecycle Security Level ($G^{\text{cos}}_{s_i}$):} The security gain of skillset $s_i$ on the CoS is formalized as the weighted activation of verification protocols during the submission and acquisition phases across all participating agents:}
    \begin{equation}
        \textcolor{blue}{G^{\text{cos}}_{s_i}(\boldsymbol{\pi}) = \sum_{j} \left( w_{\text{sub}} \cdot \pi^{(j,i)}_{\text{sub}} + w_{\text{acq}} \cdot \pi^{(j,i)}_{\text{acq}} \right),}
    \end{equation}
    \textcolor{blue}{where $w_{\text{sub}}$ and $w_{\text{acq}}$ quantify the security weight of skillset submission and acquisition, respectively.}

    \item \textcolor{blue}{\textbf{Lifecycle Resource Cost ($C^{\text{cos}}_{s_i}$):} Enforcing zero-trust identification across the skillset lifecycle directly introduces multi-dimensional communication and computational resource consumption. We formalize this cost as a joint evaluation function:}
    \begin{align}
        \textcolor{blue}{C^{\text{cos}}_{s_i}(\boldsymbol{\pi})} &\textcolor{blue}{= \sum_{j} \left[ \pi^{(j,i)}_{\text{sub}} \cdot \left( C^{\text{comm}}_{\text{sub}, j} + C^{\text{comp}}_{\text{sub}, j} \right) \right.}\nonumber\\
        & \textcolor{blue}{\left. + \pi^{(j,i)}_{\text{acq}} \cdot \left( C^{\text{comm}}_{\text{acq}, j} + C^{\text{comp}}_{\text{acq}, j} \right) \right]},
    \end{align}
    \textcolor{blue}{where $C^{\text{comm}}$ and $C^{\text{comp}}$ represent the abstract communication and computational costs mapped to agent $a_j$, respectively. Specifically, during the \textit{submission phase} ($\pi^{(j,i)}_{\text{sub}}=1$), the communication cost covers the consensus propagation overhead across the blockchain network, while the computational cost encapsulates the CPU utilization required for generating cryptographic signatures, executing verification algorithms, and running smart contracts. Conversely, during the \textit{acquisition phase} ($\pi^{(j,i)}_{\text{acq}}{=}1$), the communication cost primarily reflects the network communication delay for querying the decentralized ledger and fetching the authenticated skillset model, whereas the computational cost accounts for the CPU overhead generated by authenticating ledger states and semantic mapping algorithms.}
\end{itemize}

\subsubsection{\textcolor{blue}{Tier-2: CoC-Oriented Multi-Agent Collaboration (Three-Way Trade-off)}}
\textcolor{blue}{This tier governs the dynamic, runtime execution of multi-agent collaborative learning and inference workflows. Specifically, a \textit{task publisher} initializes a collaborative task $\mathcal{T}_k$ on the global CoS, which dynamically triggers a CoC to orchestrate and verify the zero-trust workflow. Denote $\mathcal{R}_k \subseteq \mathcal{A}$ as the universal candidate pool of available agents possessing the required skillsets to complete task $\mathcal{T}_k$, without yet considering security endorsements. To filter out latent adversarial entities, the CoC executes a admission decision policy $\pi^{(j,k)}_{\text{coc}} \in \{0, 1\}$ as a stochastic cryptographic filter, where $\pi^{(j,k)}_{\text{coc}} = 1$ mandates that agent $a_j \in \mathcal{R}_k$ successfully passes the authentication and is authorized to join the CoC for task $\mathcal{T}_k$. Consequently, the authorized agent subset is defined as $\mathcal{M}_k = \{a_j \in \mathcal{R}_k \mid \pi^{(j,k)}_{\text{coc}} = 1\}$, with its cardinality denoted as $M_k = |\mathcal{M}_k|$. In Tier-2, the dynamic orchestration exhibits a complex three-way trade-off among security, task performance, and resource cost:}

\begin{itemize}
    \item \textcolor{blue}{\textbf{Collaboration Security Level ($G^{\text{coc}}_{k}$):} Unlike conventional networks where a larger node size implies better structural path redundancy, security in cognitive multi-agent systems is strictly a decaying function of the authorized cluster size $M_k$ due to the expansion of adversarial attack surfaces, which is given by}
    \begin{equation}
       \textcolor{blue}{ G^{\text{coc}}_{\mathcal{T}_k}(M_k) = 1 - P_{\text{collusion}}\left(M_k \mid \boldsymbol{\pi}_{\text{coc}}\right),}
    \end{equation}
   \textcolor{blue}{where $P_{\text{collusion}}(\cdot)$ represents the tail probability of an adversarial majority successfully forming during decentralized consensus or distributed model aggregation. Crucially, by executing the admission policy $\boldsymbol{\pi}_{\text{coc}}$, the CoC filters out possible malicious agents, parameterized as a bound that flattens the growth curve of the collusion probability.}
    
    \item \textcolor{blue}{\textbf{Task Performance Error ($\mathcal{E}_{\mathcal{T}_k}$):} The collaborative performance of a task is evaluated by $\mathcal{E}_{\mathcal{T}_k}$, a generalized, task-specific error or performance-gap function depending on the specific task, which can be written as}
    \begin{equation}
       \textcolor{blue}{ \mathcal{E}_{\mathcal{T}_k} = \mathbb{E}\left[\mathcal{F}_k\left({\omega_{\mathcal{M}_k}}, T\right)\right] - \mathcal{F}^*_{\mathcal{R}_k},}
    \end{equation}
   \textcolor{blue}{where $\mathcal{F}_{k}(\cdot)$ represents the empirical evaluation metric of the task by the $M_k$ authorized agents and $T$ coordination iterations, and $\mathcal{F}^*_{\mathcal{R}_k}$ is the idealized theoretical performance boundary. The efficacy of task execution heavily depends on the diversity and volume of data or specialized capabilities contributed by authorized agents. Crucially, restricting $M_k$ to enhance zero-trust security may inadvertently exclude ``honest but non-conformist" edge agents from the CoC, leading to an optimality gap.}

    \item \textcolor{blue}{\textbf{Resource Cost ($C^{\text{coc}}_{\mathcal{T}_k}$):}} While CoS offers security benefits with the zero-trust principle, it also introduces a higher total resource cost due to the additional communication and computational resources required for distributed ledger verification operations. Total cost $C^{\text{coc}}_{\mathcal{T}_k}$ for a task $\mathcal{T}_k$ can be formulated as:
\begin{eqnarray}
C^{\text{coc}}_{\mathcal{T}_k} = \left[ \underbrace{\phi(M_k)}_{\text{Verification Cost}} + \underbrace{\psi(M_k, T)}_{\text{Commun/Comp Cost}} \right],
\end{eqnarray}
where $\phi(\cdot)$ represents the computational cost of the CoS executing the admission policy $\pi_{coc}$ for $M_k$ agents, and $\psi(\cdot)$ represents the communication/computation cost of multi-agent coordination across $T$ iterations.
\end{itemize}

\textcolor{blue}{In summary, this section establishes the formulated system model and the performance tradeoff for secure, ledger-anchored skillset orchestration and multi-agent collaboration in AgentNet systems. In the following section, we will introduce the architectural framework of the proposed TrustAgentNet.}


\section{TrustAgentNet Architecture}
\label{sec:framework}


\begin{figure}
\centering
\includegraphics[width=1\linewidth]{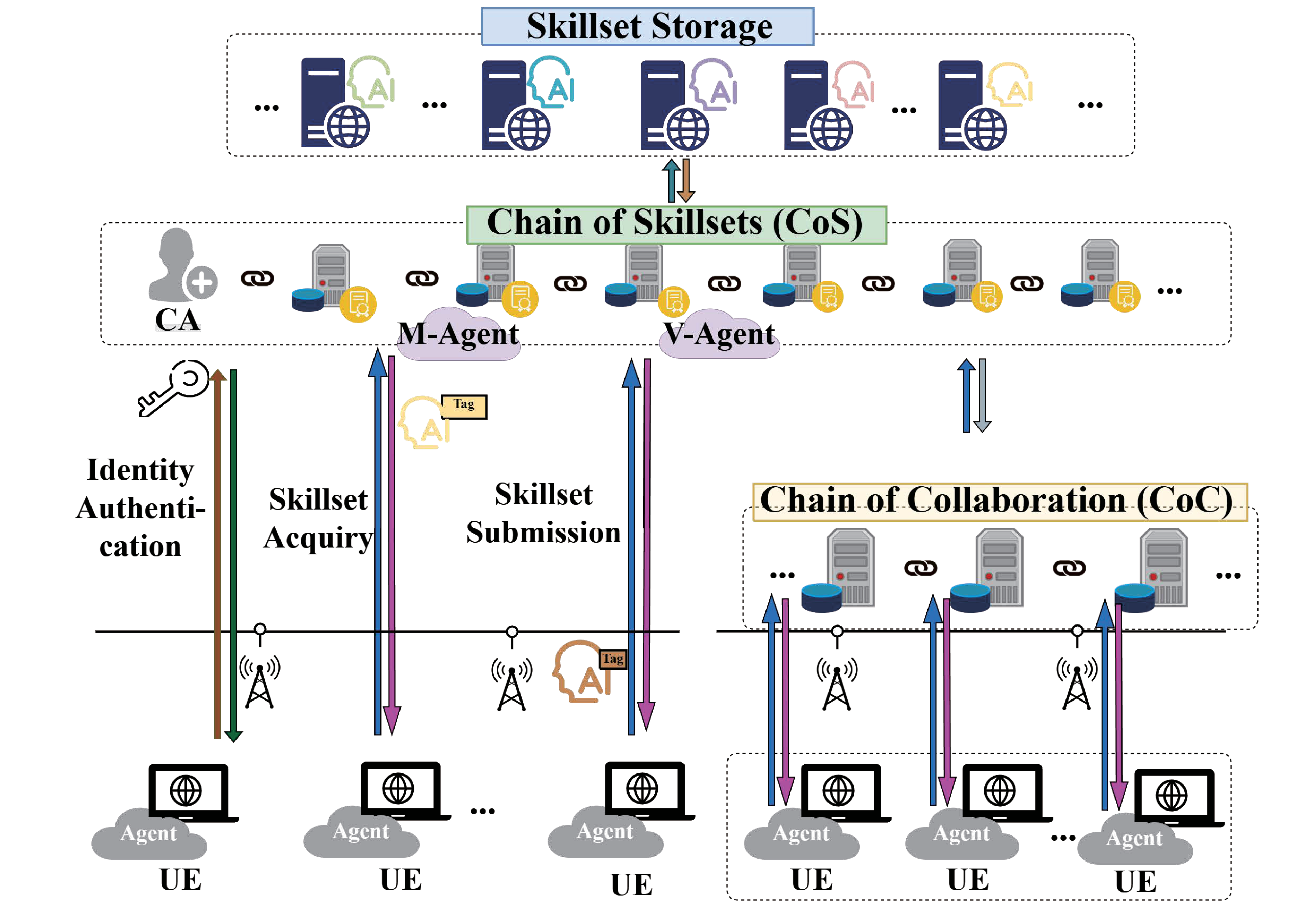}
\caption{An architectural framework of TrustAgentNet.}
\label{Figure_blckagnet_framework}
\vspace{-4mm}
\end{figure}


The architectural framework TrustAgentNet is illustrated in Fig. \ref{Figure_blckagnet_framework}, consisting of the following key components:

\begin{itemize}
     \item {\bf Infrastructure:} includes the hardware infrastructure, such as cloud and edge computing and storage resources and communication networks that connect agents and the blockchain nodes, and the software systems including high-quality datasets, accumulated skillsets, distributed ledgers and comprehensive world models available for agents' utilization.
     
     \item \textbf{Chain of skillsets (CoS)}: is the foundational, immutable ledger of TrustAgentNet, implemented as a consortium blockchain. It contains the following primary subfunctional modules: 1) \emph{Agent Identity Authentication}: provides robust digital identity and authorization services by generating and authenticating a unique, non-repudiable cryptographic Agent Identity (AID) for new agents $a_k$ joining the network; 2) \emph{Skillset Tag Registry}: stores the metadata of all available and up-to-date skillsets $s_i$, $s_i{\in}\mathcal{S}$; 3) \emph{Upload/Download Inquiry Record}: maintains the transparency and auditability of skillset usage by logging every interaction related to the submission of new skillsets (uploads) and the retrieval of existing skillsets (downloads) as a blockchain transaction. \textcolor{blue}{The detailed procedural execution and algorithm design  will be introduced in Section \ref{sec:protocol}.}

      \item \textbf{Skillset storage}: is a decentralized storage network designed to provide permanent and content-addressed persistence for all skillsets by employing decentralized storage technologies, such as the InterPlanetary File System (IPFS), utilizing a group of distributed storage resources across the network. Upon storing a skillset, the decentralized storage generates a unique, content-based identifier (CID), e.g., a Content ID in IPFS, which is recorded within the skillset's tag. Any agent can trustlessly query or download a specific skillset using the CID written in the skillset tag, guaranteeing that the retrieved content precisely matches the data initially registered and audited by the CoS. This separation of the decentralized ledger and the decentralized storage is fundamental to the scalability and trustworthiness of TrustAgentNet.

     \item \textbf{Chain of Collaboration (CoC) for a dedicated task}: facilitates privacy-preserving coordination among multiple agents to collaboratively complete a training or inference task $m$ in the CoS. \textcolor{blue}{Embedded within the broader CoS system as a runtime collaboration slice, a transient CoC is spawned from the CoS when a specific task session is initiated, and released when the task is accomplished.} The task can be initiated by a task publisher by submitting a transaction to the CoS, and a corresponding CoC is established, where agents can volunteer or be activated to join the chain. Joining an existing chain necessitates endorsement, commonly requiring a consensus of a given percentage of the current participating agents. Upon successful endorsement, each agent within the authorized active subset can submit its local outputs to the blockchain as transactions for validation and aggregation. CoS accommodates numerous parallel CoCs and supports various types of multi-agent decentralized collaboration frameworks, with each chain focusing on the collaborative learning and evolution of a specific task demand. Detailed execution and corresponding theoretical analysis will be discussed in Section \ref{sec:coc}.
         
     \item \textbf{Agents}: include various types of task-oriented agents. Each agent $a_k$ implemented in a specific environment $e$ has a unique set of locally observed dataset $\mathcal{D}_{k,e}$, and has a set of skillsets $\mathcal{S}_{k,e}$ so as to accomplish certain tasks $m{\in}\mathcal{T}_{e}$.
     
     \item \textbf{Skillset tag}: is a trusted digital certificate, endorsed by the CoS, used to advertise the capability declaration $D_i$ of a skillset, the CID linking to the actual skillset artifact stored in the Storage, and its life-cycle traceability. Typically, a tag can be a JSON metadata file, which can be queried to allow dynamic capability discovery so the network can find the most suitable set of skillsets for a given task.

\end{itemize}

\textcolor{blue}{To demonstrate how this architecture operates in practice, the subsequent sections formalize the operational workflows of TrustAgentNet into algorithmic protocols and theoretical analysis. Section \ref{sec:protocol} introduces the protocols governing the zero-trust skillset lifecycle management (Tier-1), while Section \ref{sec:coc} presents the on-chain multi-agent collaboration alongside their corresponding theoretical three-way tradeoff analysis (Tier-2).}

\section{\textcolor{blue}{Zero-Trust Skillset Lifecycle Management Protocols}}
\label{sec:protocol}

\textcolor{blue}{To realize the Tier-1 skillset management formalized in Section III-B, this section instantiates the operational workflows into two rigorous algorithmic protocols. Since standard smart contracts cannot comprehend high-level semantics or verify complex black-box intelligence, we integrate specialized LLM agents as core protocol primitives to achieve cognitive-level validation. Specifically, we first present the on-chain secure skillset acquisition and semantic mapping protocol driven by a \emph{Mapping Agent (M-Agent)}, followed by the zero-trust skillset submission and automated auditing protocol powered by a \emph{Verifying Agent (V-Agent)}.}

\subsection{\textcolor{blue}{On-Chain Secure Skillset Acquisition and Semantic Mapping Protocol}}
\label{subsec:acquisition_proto}

\textcolor{blue}{To operationalize the skillset acquisition policy ($\pi^{(j,i)}_{\text{acq}}=1$) under a strict zero-trust regime, we propose an autonomous, ledger-driven semantic-to-cryptographic translation mechanism. When an authenticated agent $a_j \in \mathcal{A}$ encounters a multi-modal environment mission goal $G_{\text{target}}$ that exceeds its current local capabilities, it signs and submits a structured acquisition request transaction $TX_{\text{acq}}$ to the CoS ledger, wrapping the raw natural language task goal $G_{\text{target}}$. The system triggers a two-phase retrieval and verification pipeline: smart-contract-mediated semantic decomposition and cryptographic integrity auditing.}

\subsubsection{\textcolor{blue}{Smart-Contract-Driven Semantic Task Decomposition via M-Agent}}
\textcolor{blue}{In heterogeneous AgentNet systems, task requirements are frequently unstructured, multi-modal, and environment-specific. To bridge this gap, a specialized M-Agent is integrated into the core network infrastructure and invoked via the CoS smart contracts as a trusted semantic resolution layer. Upon validating the signature and identity $AID_j$ of the requesting agent, the CoS smart contract triggers the M-Agent to parse $TX_{\text{acq}}$. The M-Agent decomposes the high-level goal $G_{\text{target}}$ into an optimized sequence of atomic sub-tasks and matches them against the indexed capability fields within the global immutable skillset registry $\mathcal{S}$. This ledger-mediated mapping and suite generation logic is formalized as a semantic projection function $\Phi_{\text{LLM}}(\cdot)$:}
\begin{equation}
\textcolor{blue}{\Phi_{\text{LLM}}(G_{\text{target}}) \longrightarrow \mathcal{S}^{'}_{j} := {s_{i} \in \mathcal{S} \mid \text{SemMatch}(s_{i}, }G_{\text{target}}) = 1},
\end{equation}
\textcolor{blue}{where $\mathcal{S}'_{j}$ denotes the customized minimal atomic skillset suite to accomplish the task. Through continuous contextual learning, the M-Agent can iteratively refine its mapping logic to ensure high-accuracy capability matching.}

\subsubsection{\textcolor{blue}{On-Chain Retrieval and Cryptographic Integrity Verification}}
\textcolor{blue}{Once the required skillset subset $\mathcal{S}'_{j}$ is determined, the CoS smart contract returns the corresponding secure cryptographic metadata tags for each requested $s_i \in \mathcal{S}'_{j}$ to agent $a_j$. Specifically, the returned tag for each skillset encapsulates a CID, with which as routing keys, agent $a_j$ can establish a peer-to-peer session with the decentralized storage network to fetch the original files of skillsets $\mathcal{S}'_{j}$. To eliminate intermediate tampering, man-in-the-middle poisoning, or transit data corruption, agent $a_j$ can independently recompute the cryptographic hash of the downloaded binary files, denoted as $\mathcal{H}(\text{Downloaded File})$. The downloaded skillset asset is authorized for local loading and execution if and only if it satisfies the strict equivalence check:}
  \begin{equation}
 \textcolor{blue}{   \mathcal{H}(\text{Downloaded File}) \equiv \text{CID}_{\text{CoS}},}
    \end{equation}
\textcolor{blue}{where $\text{CID}_{\text{CoS}}$ represents the authoritative hash anchored on the immutable CoS ledger. Upon successful validation of all skillsets in $\mathcal{S}'_{j}$, agent $a_j$ can update its \emph{Agent Card} to append the newly embedded capabilities, which can be subsequently advertised to agent controller or adjacent peer agents for further task orchestration.}

\subsection{\textcolor{blue}{Zero-Trust Skillset Submission and Verification Protocol}}
\textcolor{blue}{To prevent malicious, compromised, or sub-optimal agents from contaminating the decentralized registry with false capability advertisements, TrustAgentNet enforces a \emph{claim-to-capability} consistency verification protocol ($\pi^{(j,i)}_{\text{sub}}=1$) to audit the skillset's claimed capabilities via a CoS-supervised pipeline.}

\subsubsection{\textcolor{blue}{Operational Workflow}} \textcolor{blue}{Specifically, upon locally developing or acquiring a new capability, a verified agent $a_j$ encapsulates the skillset into a signed transaction $TX_{\text{sub}}$ and submits it to the CoS ledger. Once the transaction is deposited, the CoS smart contract autonomously triggers a specialized V-Agent deployed within a secure network sandbox. The V-Agent is tasked with executing a rigorous \emph{claim-to-capability consistency verification} to empirically cross-examine whether the skillset's technical implementation matches its advertised declarations. Upon terminating the verification loop, the V-Agent synthesizes and signs a comprehensive \emph{Consistency Auditing Report},  and submits it back to the CoS smart contract. If the report certifies behavioral consistency, the contract executes a \emph{Proof-of-Performance (PoP)} routine, dynamically initializes the skillset's trust score and updates the submitting agent's trust score on the global ledger.}


\begin{figure*}
\centering
\includegraphics[width=1\textwidth]{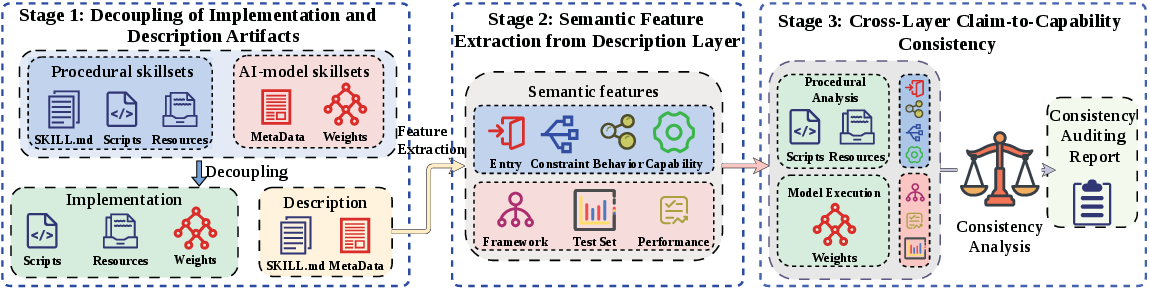}
\caption{\textcolor{blue}{Graphic illustration of the V-Agent consistency verification protocol for different types of skillset submissions.}}
\label{Figure_v_agent_flow}
\vspace{-4mm}
\end{figure*}

\subsubsection{\textcolor{blue}{Verification Protocol Design}}

\textcolor{blue}{Fig. \ref{Figure_v_agent_flow} illustrates the verification protocol orchestrated by the V-Agent. Crucially, the verification protocol mainly includes three stages, accommodating both \emph{AI-model-based skillsets} (comprising neural network weights/gradients, and model metadata) and \emph{procedural skillsets} (comprising structured execution scripts and code manifests compliant with the ClawHub):}

\textcolor{blue}{\noindent \textbf{Stage 1: Decoupling of implementation and declaration artifacts.} Upon intercepting $TX_{\text{sub}}$, the V-Agent isolates the raw submission data based on its architectural track. For \emph{AI-model skillsets}, it separates the weight parameters and gradient vectors from the behavioral metadata files. For \emph{procedural skillsets}, it decouples the execution scripts and runtime resources from the declarative documentation (\texttt{SKILL.md}). This separation yields two distinct vectors for downstream evaluation: the operational \emph{Implementation Layer} (executable code or model weights) and the semantic \emph{Declaration Layer} (natural language claimed manifests and target benchmarks in the skillset tag).}

\textcolor{blue}{\noindent \textbf{Stage 2: Semantic feature extraction from declaration layer.} The V-Agent then extracts the skillset's claimed capability profile from the declaration artifacts, formalizing it as a multi-dimensional semantic feature tensor $\mathcal{F}_{\text{claim}}$. This extraction conforms to the dual-track skillset architecture:
\begin{itemize}
    \item \emph{AI-Model skillsets}: The extracted features map the structural metadata required for neural deployment, specifically the deep learning framework, the designated test set, and the claimed performance metrics.
    \item \emph{Procedural skillsets}: The extraction targets four mutually exclusive runtime units from the documentation (\texttt{SKILL.md}) to define execution boundaries: \emph{Entry} points (invocation triggers/APIs), \emph{Constraints} (hard external dependencies/version requirements), \emph{Capability} profiles (functional utility statements), and \emph{Behavior} protocols (internal data routing/fallback logic).
\end{itemize}}

\textcolor{blue}{\noindent \textbf{Stage 3: Cross-layer claim-to-capability consistency analysis.} The LLM-based V-Agent completes the verification loop by orchestrating automated, sandboxed executions to cross-examine $\mathcal{F}_{\text{claim}}$ against the actual implementation layer. The evaluation mechanics are tailored to the skillset track:
\begin{itemize}
\item \emph{AI-Model consistency verification}: The V-Agent deploys a pure sandbox environment, mounts the raw submitted model weights, and executes batch inference over the declared test set. The system evaluates consistency by directly checking if the empirically measured performance matches the claimed metric within a tolerance, i.e., $\Delta_{\text{p}} = | \mathcal{P}_{\text{claimed}} - \mathcal{P}_{\text{empirical}} |<\epsilon$.
\item \emph{Procedural consistency verification}: Due to the programmatic complexity of execution scripts, the V-Agent evaluates each extracted procedural feature $f$ via an audit function $\Psi(f)$ to render three distinct architectural verdicts:
The V-Agent then conducts a formal consistency analysis, contrasting the implemented scripts/resources against the claimed features anchored in $\mathcal{F}_{\text{claim}}$. The consistency validation rule is formulated as:
\begin{equation}
\Psi(f) = \begin{cases}
\text{Match}, & \text{if $f$ is verified} \\
\text{Miss}, & \text{if $f$ is unimplemented} \\
\text{Uncertain}, & \text{if $f$ depends on the framework.}
\end{cases}
\end{equation}
Here the \emph{Uncertain} verdict is assigned to features natively bound to framework prompt spaces (e.g., OpenClaw core orchestrations) or external APIs, flagging them for further determination.
\end{itemize}
}
    
\subsubsection{\textcolor{blue}{Dual-Tier Trust Algorithm Design}}
\textcolor{blue}{To systematically govern submission compliance across heterogeneous asset tracks, the CoS smart contract enforces a dual-tier trust score algorithm that parallelly tracks a Skillset Trust Score $\mathcal{TS}_s \in [0,1]$ and an Agent Trust Score $\mathcal{TS}_a \in [0,1]$. The individual skillset trust score $\mathcal{TS}_s \in [0,1]$ is given by
\begin{equation}
\mathcal{TS}_s = \begin{cases}
\exp(-\gamma \cdot \Delta{\text{p}}), & \text{for AI-Model Skillsets} \\
\frac{1}{F}\sum_{i=1}^F ( \Psi(f_i) ), & \text{for Procedural Skillsets}
\end{cases}
\end{equation}
where $\Delta_{p}$ quantifies the sandbox performance deviation for AI-model skillset with a tuning parameter $\gamma > 0$; and $\Psi(\cdot) \in [0, 1]$ represents a predefined monotone weighting function mapping the discrete architectural verdicts ($\text{Match}, \text{Uncertain}, \text{Miss}$) to bounded continuous scores. Concurrently, the presenting agent's global trust score $\mathcal{TS}_a$ is defined to track long-term historical compliance via an exponential moving average
\begin{equation}
\mathcal{TS}_a^{(t)} = \rho \mathcal{TS}_a^{(t-1)} + (1-\rho)\mathcal{TS}_s,
\end{equation}
where $\rho$ denotes the credit inertia factor. Multi-tier governance policies can further be continuously executed based on this dual-tier matrix: if the skillset score falls below the admission threshold ($\mathcal{TS}_s < \theta_{\text{adm}}$), the specific skillset is rejected from CoS; if the aggregated trust score of an agent drops below the revocation threshold ($\mathcal{TS}_a < \theta_{\text{rev}}$), the smart contract can trigger a predefined punishment routine, e.g., executing cryptographic access revocation by blacklisting the agent's public key, or an economic slashing penalty that forfeits its pre-staked computational credits.}
\section{CoC-Oriented On-Chain Multi-Agent Collaboration and Theoretical Analysis} 
\label{sec:coc}
\textcolor{blue}{This section presents the on-chain multi-agent collaborative learning and inference sub-architecture within TrustAgentNet, designed to facilitate privacy-preserving, decentralized coordination across heterogeneous intelligence entities. Since individual agents operate in distinct, potentially antagonistic environments, they cannot directly expose their local datasets; thus, multi-agent collaboration must rely exclusively on the secure exchange of intermediate learning parameters, such as model output embeddings and cryptographic representations of gradients. To operationalize this environment, a \emph{task publisher} initializes collaborative training or inference missions on the global CoS, which dynamically triggers the instantiation of an isolated CoC to orchestrate and verify the zero-trust workflow. In the remainder of this section, we first elaborate on three specialized collaborative framework primitives, including federated global model training, distributed MoE, and multi-agent model partitioning and sharing (MoPS), followed by a rigorous theoretical analysis on the three-way tradeoff of security, task performance and resource cost.}

\begin{figure*}
		\centering
		\subfloat[]{ \includegraphics[height=1.8in]{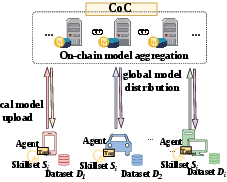}} 
		\subfloat[]{ \includegraphics[height=1.8in]{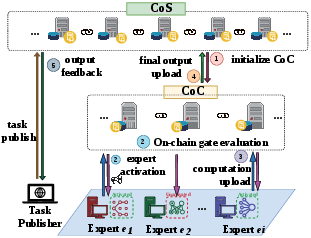}}
		\subfloat[]{ \includegraphics[height=1.8in]{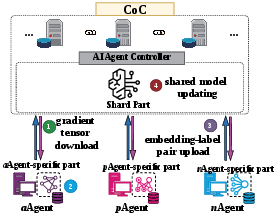}}
		\caption{On-chain multi-agent collaborative training/inference procedures of (a) federated global model training and (b) distributed mixture of experts networks, and (c) multi-agent model partition and sharing.}
		\label{Fig_MultiAgentColl}
        \vspace{-4mm}
	\end{figure*}
\subsection{\textcolor{blue}{On-Chain Multi-Agent Collaborative Learning/Inference}}

The CoC framework specifically accommodates a series of multi-agent collaboration tasks as follows:

1) \textbf{Federated global model training frameworks}: This involves the collaborative training of a single global model among \textcolor{blue}{a selected set of authenticated agents $\mathcal{M}_i \subseteq \mathcal{A}$ possessing an identical skillset $s_i$ with their private local datasets $\{\mathcal{D}_j\}_{j \in \mathcal{M}_i}$}, as illustrated in Fig. \ref{Fig_MultiAgentColl}(a). During the $t$-th round of coordination, the
following steps are sequentially executed by the agents and the CoC, maintaining a rigorous zero-trust posture across all data exchanges.

\begin{itemize}
    \item[(1)] {\em Local model updating (at each agent)}: \textcolor{blue}{Each authorized agent $a_j \in \mathcal{M}_i$ queries the CoC to retrieve the prior consensus-validated global parameters $w_{\text{g}}^{(t-1)}$, and conducts local training utilizing its local dataset $\mathcal{D}_j$ for $E$ rounds to yield updated parameters $w_j^{(t)}$}. 

\item[(2)] {\em Parameter uploading to the chain}: After local updating, each agent submits its local model parameters to the CoC as a \textcolor{blue}{cryptographically signed blockchain transaction $TX_j^{(t)} := \langle \text{AID}_j, t, w_j^{(t)} \rangle$ to guarantee immutable non-repudiation.} Adhering to the zero-trust principle, $TX_j^{(t)}$ is broadcast for endorsement and consensus validation before including it in a proposed block.

\item[(3)] {\em Smart-Contract-Mediated Model aggregation (at the chain)}: The core function of aggregating the local updates occurs on the chain, mediated by a specialized Aggregation Smart Contract (ASC). \textcolor{blue}{Upon consensus commitment, the ASC is automatically triggered to execute the global parameter aggregation $w_{\text{g}}^{(t)} = \Psi_{\text{ASC}}\left( \{TX_j^{(t)}\}_{j \in \mathcal{M}_i} \right)$, where $\Psi_{\text{ASC}}(\cdot)$ represents the aggregation logic. The synthesized global state $w_{\text{g}}^{(t)}$ is securely committed to the CoC state, which is transparently queryable for agents to securely initialize round $t+1$.}

\end{itemize}



2) \textbf{Distributed MoE frameworks of large models}: Distributed MoE supports a group of specialized neural networks (called experts) \textcolor{blue}{deployed across heterogeneous edge devices, dynamically activated by a gating network to collaboratively execute complex tasks \cite{Wang2025}}, as illustrated in Fig. \ref{Fig_MultiAgentColl}(b). Upon task instantiation, the workflow for on-chain distributed MoE inference is as follows:

\begin{itemize}
    \item[(1)] {\em On-chain gate evaluation (at the chain)}: \textcolor{blue}{Instead of relying on a centralized coordinator, the on-chain gating network is executed entirely within a specialized routing smart contract on the CoC. For instance, the contract computes the gating weight vector $g$ and dynamically selects the optimal subset of activated experts $\mathcal{E}_{\text{act}}$ via a top-$K$ sparse activation mechanism\cite{ShazeerMMDLHD17}. }

\item[(2)] {\em Expert local computation (at each expert)}: The edge users then download the task, utilize the activated experts to process the learning task, and submit \textcolor{blue}{their mid-layer computational outputs into a cryptographically signed transaction $TX_e := \langle \text{AID}_e, \text{TaskID}, y_e \rangle$ to the CoC} similar to the previous case. 

\item[(3)] {\em Smart-Contract-Mediated Output aggregation (at the chain)}: Similarly, the submitted expert outputs $\{y_e\}_{e \in \mathcal{E}_{\text{act}}}$ are aggregated by a specialized ASC, and the final outputs $Y$ of MoE are then \textcolor{blue}{committed to the global CoS state through cross-chain consensus registry, achieving verifiable end-to-end task fulfillment for secure retrieval by the task publisher.}



\end{itemize}

3) \textbf{Multi-agent model partitioning and sharing frameworks}: \textcolor{blue}{To adapt large models to edge resource constraints, we integrate our recently proposed MoPS framework \cite{Xiao2025SANet} into the CoC subsystem, segmenting deep networks into a global foundation {\em shared-part} and task-specific {\em agent-specific parts}, as illustrated in Fig. \ref{Fig_MultiAgentColl}(c). More specifically, in MoPS, the agent controller hosts the shared parameters $\omega^{\text{sh}}_{m,t}$ while assigning localized slices $\omega^{i}_{m,t}$ to individual agents. To securely synchronize the collaborative training loop without central single-point-of-failure or privacy risks,} during the $t$-th round of coordination, the following steps are sequentially executed and repeated by the agents and the agent controller:  

\begin{itemize}
    \item[(1)] {\em Local model updating (at each agent):} Each agent $a_i$ first downloads the consensused gradient tensor $g^{sh,i}_{m,t-1}$ from the agent controller via the CoC. It then computes and updates the gradient of its local model parameters by following Eqs. (8)-(9) in \cite{Xiao2025SANet}.

\item[(2)] {\em Embedding uploading to the chain: } Each agent $a_i$ calculates the model output embedding $z^i_{t} = f_{\btomega^i_{i, t}}(x^i_{t})$ and then uploads the embedding-label pair $\langle z^i_t, y^i_t \rangle$
to the CoC as a signed transaction for consensus.

\item[(3)] {\em Shared model updating (at the chain):} Upon receiving the embedding-label pairs from all called agents, the shared-part model $\btomega^{sh}_{i, t}$ is updated by performing the aggregated gradient descent using a specialized ASC and the output gradient tensor $g^{sh,i}_{m,t}$ is submitted to CoC as a block, establishing an audited baseline for agents to safely initialize iteration $t+1$.

\end{itemize}

\textcolor{blue}{4) \textbf{Discussions on scalability and resource orchestration:}
A potential concern in large-scale AgentNet deployment is the management overhead of numerous concurrent CoCs. To mitigate this, TrustAgentNet employs an ephemeral instantiation strategy. Unlike persistent ledgers, each CoC is a lightweight, on-demand coordination instance. Upon task fulfillment, the CoC undergoes a state-compacting process where only the final verified results and audit logs are synchronized to the CoS, followed by the complete release of the CoC’s runtime resources. Furthermore, an adaptive on-chain/off-chain mode selection can be introduced to achieve better tradeoff between security and resource consumption. For routine tasks in trusted environments, agents can perform off-chain collaborative inference to minimize latency and overhead. The resource-intensive on-chain coordination is dynamically activated only when potential security risks or malicious anomalies are detected. This hierarchical and adaptive approach ensures that the systemic overhead remains linear relative to the task density, rather than experiencing the combinatorial explosion typically associated with flat, non-partitioned blockchain architectures.}

\subsection{Theoretical Analysis on Three-way Tradeoff Among Security Level, Agent Performance, and Resource Consumption}
As described in Section \ref{sec:coc}-A, TrustAgentNet enables blockchain-secured multi-agent collaborative learning for distinct frameworks via dedicated CoCs. This subsection further presents the theoretical analysis of three-way tradeoff between security, agent skillset performance, and resource consumption in two representative frameworks: federated global modeling training and multi-agent model partitioning and sharing. 

\subsubsection{Theoretical Results for Federated Global Model Training Frameworks}
For a specific skillset $s_i \in \mathcal{S} := \bigcup_{a_k \in \mathcal{A}} \mathcal{S}_k$, a subset of agents $\mathcal{M}_i \subseteq \mathcal{R}_i$ that possess this skillset and obtain the consensus of the current participating agents will spontaneously engage in the co-training, which can be formulated as:
\begin{eqnarray}
    \mbox{\bf P1:} \ \underset{w_i \in \mathbb{R}^d}{\min} \ F_i(w_i) := \sum_{a_k \in \mathcal{M}_i} p_{i, k} F_{i, k}(w_i),
\end{eqnarray}
where $F_{i, k}(w_i) := \frac{1}{|\mathcal{D}_k|}\sum_{(x_{k, j}, y_{k, j}) \in \mathcal{D}_k} f(w_i; x_{k, j}, y_{k, j})$ represents the local objective function for agent $a_k \in \mathcal{M}_i$, and $p_{i, k}$ denotes the weight of agent $a_k$ for skillset $s_i$, where $0 \le p_{i, k} \le 1$ and $\sum_{a_k \in \mathcal{M}_i} p_{i, k} = 1$. We define $f: \mathbb{R}^d \rightarrow \mathbb{R}^{+}$ as the non-negative loss function reflecting the error of the model $w_i$ evaluated on sample $(x_{k,j}, y_{k,j})$. 

Suppose the optimal model for skillset $s_i$ is given by $w_{\mathcal{R}_i}^* = \arg\min_{w_i} \sum_{a_k \in \mathcal{R}_i}p_{i, k}F_{i, k}(w_i)$ where $\mathcal{R}_i$ denotes the universal set of agents that have skillset $s_i$ without considering security endorsement. We then have the following definition on the impact of the introduction of security authentication.

\begin{definition}
    The impact of removing a subset $\mathcal{Q}_i = \mathcal{R}_i \setminus \mathcal{M}_i$ of agents on skillset $s_i$ due to their failure to meet security certifications can be defined as the difference between the global optimal performance with all $R_i$ agents and that with only the $M_i$ certified agents. This discrepancy is given by:
    \begin{eqnarray}
        C_{\mathcal{M}_i}^*(F_i) = \sum_{a_k \in \mathcal{R}_i} p_{i, k}(F_{i, k}(w_{\mathcal{M}_i}^*) - F_{i, k}(w_{\mathcal{R}_i}^*)).
    \end{eqnarray}
\end{definition}

We can then derive the following theoretical bound of the skillset performance with security endorsement on CoC-based federated global model training frameworks.

\begin{theorem}\label{thm1}
    Suppose the following assumptions hold:
    \begin{assumption}\label{smooth}
        The objective function $F_{i,k}(w)$ is $L$-smooth for any skillset $s_i$ and agent $a_k$, i.e., $\|F_{i,k}(w) - F_{i,k}(w')\| \le L \| w-w' \|$.
    \end{assumption}
    \begin{assumption}\label{convex}
        The objective function $F_{i,k}(w)$ is $\mu$-convex for any skillset $s_i$ and agent $a_k$, i.e., $\|F_{i,k}(w) - F_{i,k}(w')\| \ge \mu \| w-w' \|$.
    \end{assumption}
    \begin{assumption}\label{gradient_bounded}
        The stochastic gradient of the loss function $F_{i,k}(w)$ is upper bounded for any skillset $s_i$ and agent $a_k$, i.e., $ \mathbb{E}\|\nabla F_{i,k}(w)\| \le G$.
    \end{assumption}
    \begin{assumption}\label{variance_bounded}
        The stochastic gradient of the loss function $F_{i,k}(w)$ is variance-bounded for any skillset $s_i$ and agent $a_k$, i.e., $ \mathbb{E}\|\nabla F_{i,k}(w) - \mathbb{E}[\nabla F_{i,k}(w)]\| \le \sigma_k$.
    \end{assumption}
    Then, with $\kappa = \frac{L}{\mu}, \gamma=\max\{\frac{8L}{\mu}, E\}$, and the learning rate $\beta_t = \frac{2}{\mu (\gamma + t)}$, we have
    \begin{equation}
    \begin{aligned}
        \mathcal{E}_{s_i} \le &\frac{4 \kappa}{\mu(\gamma+ET)} \!\left(\!\sum_{a_k \in \mathcal{R}_i} \! \frac{p_{i,k}^2 \sigma_k^2}{D_k} \!+\! 8E^2 G^2 \!+\! A_{\mathcal{M}_i} \!+\! 6L C_{\mathcal{M}_i}^*(F_i)\!\!\right) \\\nonumber
        &+ L B_{\mathcal{M}_i}, \nonumber
    \end{aligned}
    \end{equation}
    where $D_k = |\mathcal{D}_k|$ is the cardinality of the set $\mathcal{D}_k$, $A_{\mathcal{M}_i} =  \frac{\mu^2(\gamma+1)}{4} \lVert w_0 - w_{\mathcal{M}_i}^* \rVert^2$, and $B_{\mathcal{M}_i} = \lVert w_{\mathcal{M}_i}^* - w_{\mathcal{R}_i}^* \rVert^2$.
\end{theorem}
\begin{IEEEproof}
    See Appendix \ref{proof_of_thm1}.
\end{IEEEproof}
\textcolor{blue}{
\begin{remark}
    Assumptions \ref{smooth}-\ref{variance_bounded} are commonly introduced for most theoretical analyses of gradient-based model training processes \cite{ling2024convergence, chang2023federated, bukharin2023robust}. In particular, Assumptions \ref{smooth} and \ref{convex} guarantee that the objective function's rate of change is bounded, thereby ensuring that infinitesimal adjustments to model weights yield predictable and manageable shifts in the loss landscape. Assumption \ref{gradient_bounded} and \ref{variance_bounded} characterize the stability and reliability of the stochastic optimization process. Specifically, Assumption \ref{gradient_bounded} ensures that the magnitude of the stochastic gradients remains within a reasonable range, preventing the model from experiencing gradient explosion during the training process. In practice, this assumption is often enforced through the use of gradient clipping techniques \cite{zhang2019gradient}. Assumption \ref{variance_bounded} restricts the noise level inherent in stochastic gradient estimation, implying that the updates from individual agents or data batches do not deviate excessively from the true gradient direction.
\end{remark}}

\textcolor{blue}{\begin{remark}
    Theorem 1 formalizes how security authentication non-linearly impacts the multi-agent optimization process. Specifically, a stricter authentication threshold inherently filters out suspicious agents, reducing the cardinality of the active collaborative agent cluster. From an optimization perspective, losing these agents, especially those possessing unique skillset characteristics, directly scales the statistical heterogeneity and amplifies the stochastic gradient variance, which is mathematically mapped onto the inflation of the structural bounds $B_{\mathcal{M}_i}$ and $C_{\mathcal{M}_i}^*$. Rather than a simple linear tradeoff, this formulation uncovers a delicate coupling: while elevating the security threshold secures the ecosystem, it introduces an implicit verification noise that widens the optimality gap and may destabilize the convergence of skillsets.
\end{remark}}

\subsubsection{Theoretical Results for Multi-agent Model Partitioning and Sharing Frameworks}

Note that, different from the problem {\bf P1} focusing on finding the global optimization solution to minimize a single objective function, the optimization objective in multi-agent model partitioning and sharing frameworks involves a vector of objectives from different agents. In this case, it is generally impossible to find a single global optimal solution that minimizes the loss functions of all the agents. To address this issue, in this paper, we consider the Pareto optimal solution, a metric commonly adopted for evaluating the tradeoff among a collection of possibly conflicting goals and objectives. Let $\bOmega_i = \langle w_{i, sh}, w_{i, 1}, ..., w_{i, M_i} \rangle$ denote the overall model for the skillset $s_i$, where $w_{i, sh}$ denotes the shared part of the model and $w_{i, k}$ for $a_k \in \cM_i$ denotes the local part of the model. More formally, we define the Pareto optimal solution for the multi-agent system as follows:
\begin{definition}
A solution profile $\bOmega_i$ is called {\em Pareto stationary solution} if there exists a set of non-negative weights $p_{i, k}$ for $a_k \in \cM_i$ such that $\sum_{a_k \in \cM_i} p_{i, k} = 1$ and $\sum_{a_k \in \cM_i} p_{i,k} \nabla F_{i, k} (\bOmega_i, \cD_k) = 0$. A solution profile $\bOmega_i^*$ is {\em Pareto optimal} if no other Pareto stationary solution $\bOmega_i$ for $\bOmega_i \neq \bOmega_i^*$ such that $F_{i, k} (\bOmega_i) \le F_{i, k} (\bOmega_i^*)$ for all $a_k \in \cM_i$ and $F_{i, k} (\bOmega_i) < F_{i, k} (\bOmega_i^*)$ for at least one $a_k \in \cM_i$.
\end{definition}

From the above definition, we can observe that the Pareto stationary solution seeks a feasible solution set in which no other feasible solution can improve one agent's objective without causing a deterioration in at least one other agent's objective. Therefore, the optimization problem {\bf P1} can then be rewritten in the following form:
\begin{eqnarray}
\mbox{\bf P2:} \ \min_{\langle w_{i, sh}, w_{i, 1}, ..., w_{i, M_i} \rangle} \| \sum_{a_k \in \cM_i} p_{i,k} \nabla F_{i,k} (\langle w_{i,sh}, w_{i,k} \rangle) \|^2.
\label{eq_equivParetoStation}
\end{eqnarray}

Let us now present the following theoretical bound of the skillset performance with security endorsement on blockchain based on multi-agent model partitioning and sharing framework.

\begin{theorem}\label{thm2}
    Suppose assumption \ref{smooth} holds, and the following assumption holds: 
    \begin{assumption}
        For any $p_{i,k}, a_k \in \cM_i$, the initialized model $\bOmega_0$ satisfies that 
        \begin{equation}
        \mathbb{E}[\sum\limits_{a_k \in \cM_i} p_{i,k} F_{i,k}(\bOmega_0)] - \min\limits_{\bOmega_i} \mathbb{E}[ \sum\limits_{a_k \in \cM_i} p_{i,k} F_{i,k} (\bOmega_{\cM_i})] \le C_{I}. \nonumber
        \end{equation}
    \end{assumption}
    Suppose the model $\{ \bOmega_{i,t} \}_{t=1}^T$ is trained by the static weighting $p_{i, k}$ with $\beta_t = \beta \le \frac{1}{2 L}$. We can prove the following result: 
    \begin{equation}
        \frac{1}{T} \sum_{t=0}^{T-1}\mathbb{E} [\cE_{s_i, t}] \le \sqrt{\frac{2(C_I + C_{\cM_i}(F_i))}{(2\beta-L\beta^2)T} }.
    \end{equation}
    \begin{IEEEproof}
        See Appendix \ref{proof_of_thm2}.
    \end{IEEEproof}

    \begin{remark}
        Theorem \ref{thm2} demonstrates that the average expected error decays at a rate of $\cO(\sqrt{1/T})$, implying that increasing the number of coordination rounds $T$ reduces the optimization error. However, the total resource cost, including communication, computation, and endorsement/verification overhead, scales linearly with $T$. Moreover, the bias term $C_{\cM_i}(F_i)$ captures the performance degradation caused by excluding uncertified agents, which remains fixed regardless of $T$. Together, these factors reveal a fundamental three-way trade-off among accuracy, resource expenditure, and security-induced bias in certified multi-agent model partitioning and sharing frameworks.
    \end{remark}
    
\end{theorem}

\section{Prototype and Experimental Results}
\label{sec:evaluation}
In this section, we introduce our developed TrustAgentNet prototype and present experimental results under various scenarios. 

\subsection{Prototype and Experimental Setup}

We develop a TrustAgentNet prototype based on an open-source Radio Access Network (RAN) and softwareized 5G core network, as shown in Fig. \ref{Figure_Prototype}, \textcolor{blue}{which extends our previous AgentNet platform developed in \cite{Xiao2025SANet} by incorporating the proposed blockchain-secured zero-trust layers. While the primary focus of this work is on security, the integration of the RAN and 5GC allows us to evaluate the framework's performance within a physically distributed and realistic wireless networking environment.}

\begin{figure}[t]
\centering
\includegraphics[width=1\linewidth]{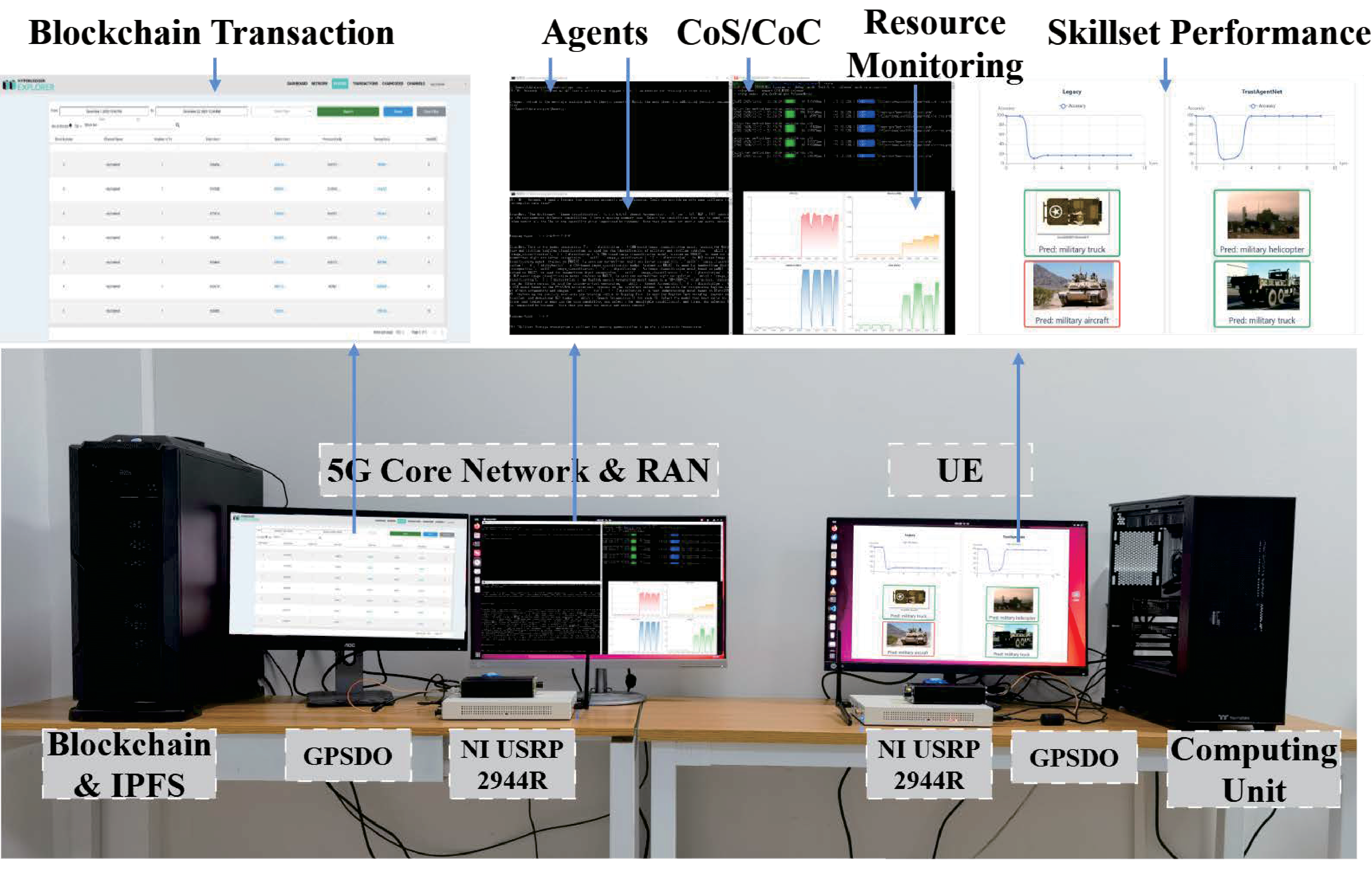}
\caption{TrustAgentNet prototype.}
\label{Figure_Prototype}
\end{figure}

\noindent{\bf Hardware \textcolor{blue}{and Network Topology}:} \textcolor{blue}{To simulate a tiered 6G AgentNet, the prototype components are distributed across a hierarchical topology.} The TrustAgentNet prototype hardware consists of five major components: the gNodeB (gNB), the user equipment (UE), the 5G core (5GC) network, the blockchain server, and the skillset storage server, as illustrated in Fig. \ref{Figure_Prototype}. Specifically, the gNB--which handles all RAN functions--is instantiated using the srsRAN open-source software suite \cite{srsRAN19}. This software-defined gNB connects to the UE via a Universal Software Radio Peripheral (USRP), NI USRP 2944R, which provides the essential hardware interface for over-the-air communication. The 5GC is implemented entirely using the Open5GS project installed on a workstation equipped with an Intel(R) Core(TM) i9-13900K CPU@5.8GHz, 128GB of DDR5 RAM, and an NVIDIA GeForce RTX 4090 GPU. We utilize a desktop server equipped with an Intel i5-10500T CPU and 64GB of memory to host the blockchain nodes and the skillset storage. 



\noindent{\bf Blockchain:} We adopt the open source Hyperledger Fabric\cite{androulaki2018hyperledger} as the consortium blockchain platform for both the CoS and each CoC. \textcolor{blue}{Specifically, the blockchain and storage infrastructure is containerized via Docker. The storage layer utilizes three decentralized IPFS nodes running Kubo v0.40.1. The blockchain network is configured with four peer nodes, one orderer node, and four independent CouchDB instances serving as state databases.} Etcdraft and Raft are adopted as the consensus protocol and algorithm, respectively. \textcolor{blue}{Smart contracts are instantiated via dedicated chaincode containers to dynamically execute the zero-trust skillset acquisition and submission workflows. Specifically, CoS deploys a set of generic key-value smart contracts on Hyperledger Fabric, utilizing the unique cryptographic content hash of the skillset metadata as the immutable primary key to anchor the entire lifecycle states. Furthermore, the upper-layer business gateway encapsulates Fabric's multi-phase transaction protocols via automated serialization and type routing mechanisms, allowing distributed agents to transparently invoke atomic ledger operations. During the skillset submission pipeline, the smart contracts persistently anchor critical milestones including initial report ingestion, CID cryptographic updates, and the final verification report anchoring; during the skillset acquisition pipeline, the smart contracts orchestrate task initialization, batch queries of model metadata, filtering result anchoring, and data retrieval, thereby constructing an immutable, end-to-end on-chain data pipeline.}

\noindent{\bf Agents:} We implement DeepSeek-v3 to function as both the M-Agent and V-Agent deployed at the 5GC. Another LLM-based specialized agent, Security-agent (S-Agent), is further implemented at the UE side \textcolor{blue}{to guarantee agents' runtime execution security. The S-Agent continuously monitors operational integrity and automatically initiates the security enhancement process by leveraging the CoS if a severe degradation in skillset performance is detected.}

\noindent{\bf Skillsets:} We implement a diverse set of skillsets \textcolor{blue}{to thoroughly validate both the semantic mapping capability of the M-Agent and the behavioral integrity verification performance of the V-Agent. Firstly, the global CoS ledger is pre-populated with the metadata profiles of $200$ heterogeneous AI model skillsets spanning $12$ diverse application domains including medical diagnostics, autonomous driving, speech processing, and natural language processing (NLP). Concurrently, to benchmark the cross-modal zero-trust consistency verification capability of the V-Agent, we construct two distinct adversarial datasets for AI-model and procedural skillsets. Specifically, we construct an AI-model dataset comprising $50$ distinct AI model skillsets generated by cross-pairing $7$ benchmark datasets (MNIST, FashionMNIST, EMNIST, CIFAR-10, CIFAR-100, SVHN, and Vehicle) with $5$ representative neural network architectures (LeNet, SimpleCNN, MLP, ResNet18, and VGG16), within which $40$ honest skills are evaluated against $10$ meticulously designed adversarial skills injecting diverse deception vectors such as fabricated inference accuracy, descriptive metadata discrepancies, structural architecture spoofing, dataset origin falsification, poisoned model weights, and multi-dimensional blended deceptions. Moreover, the procedural skillset repository comprises 171 programmatic skills extracted from the ClawHub platform (with over 50 stars), covering data analysis, financial trading, content generation, and system tools.}


\subsection{\textcolor{blue}{Evaluation of CoS Lifecycle Operations}}
\textcolor{blue}{In this subsection, we focus on evaluating the performance and security resilience of Tier-1 lifecycle operations on the CoS. We first quantify the architectural overheads introduced by zero-trust skillset acquisition and skillset submission, and then present a generalized case study demonstrating the framework's autonomous self-healing capabilities against malicious attacks.}

\subsubsection{\textcolor{blue}{Evaluation of On-Chain Skillset Acquisition and Semantic Mapping}}

\textcolor{blue}{Fig. \ref{fig:acq_compa} compares the multi-dimensional resource consumption of TrustAgentNet against a trust-by-default \textit{w/o-CoS} baseline during the skillset acquisition pipeline, including the measured execution latency (ms), CPU utilization time, and traffic volume (KB) under  varying global skillset repository sizes $|\mathcal{S}|$. Specifically, for CoS, we consider the four sequential operational phases: on-chain request submission ($P_1$), off-chain M-Agent sandboxed inference ($P_2$), on-chain result anchoring ($P_3$), and CoS metadata delivery ($P_4$). For the \textit{w/o-CoS} baseline, agents directly query the M-Agent and retrieve feedback, completely bypassing decentralized cryptographic verification and consensus synchronization. For each skillset size $|\mathcal{S}|$, the agent issues an identical natural-language request, and each metric is evaluated by averaging the outcomes of 20 independent experimental trials.}

\textcolor{blue}{It can be clearly seen that the zero-trust overhead introduced by CoS is marginal. Across all resource dimensions, consumption is heavily dominated by the off-chain M-Agent inference ($P_2$) due to LLM cognitive mapping workloads. Conversely, blockchain-anchored phases ($P_1, P_3, P_4$) incur minor costs, validating that our blockchain-IPFS synergy and on/off-chain integration successfully confines on-chain operations to lightweight metadata hashing without processing bottlenecks. Specifically, the blockchain transactions executed during the on-chain request submission ($P_1$) and result anchoring ($P_3$) phases introduce a latency overhead of approx. $15$ ms, which is attributed to the endorsement and ordering consensus mechanisms. In terms of computational overhead, $P_2$ in CoS consumes more CPU utilization than the baseline due to on-chain metadata queries. Moreover, the framework exhibits robust structural scalability as the asset inventory $|\mathcal{S}|$ scales from $30$ to $200$. As $|\mathcal{S}|$ expands, the resource costs of the ledger-based phases ($P_1, P_3, P_4$) remain invariant, while that of $P_2$  scales sub-linearly to the semantic search space, demonstrating that TrustAgentNet can accommodate AgentNet systems with large-scale skillset inventories.}

\begin{figure*}[htbp]
\centering
\includegraphics[width=0.8\linewidth]{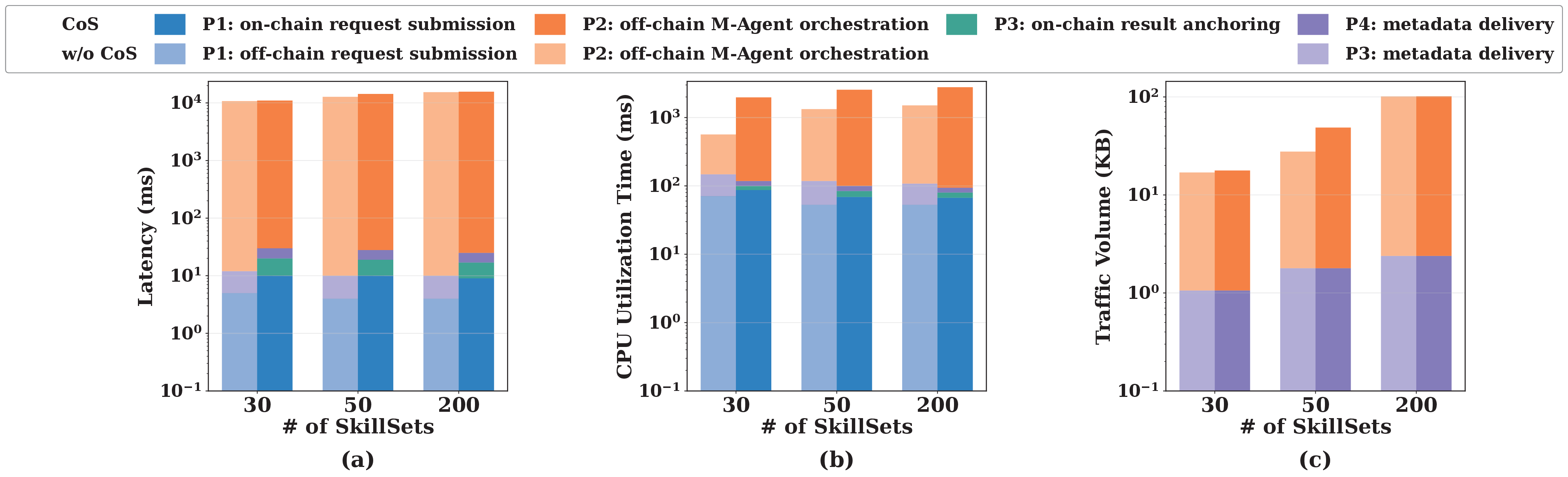}
\caption{\textcolor{blue}{Comparison in (a) latency, (b) CPU utilization time and (c) traffic volume between the proposed CoS and the baseline (w/o CoS) during the skillset acquisition process across varying scales of skillsets in the system. Note that to eliminate visual masking effects caused by multi-magnitude scaling disparities, the dominant $P_2$ (M-Agent inference) phase is placed at the top of the stacking hierarchy to preserve the visual legibility of the remaining micro-overhead stages.}}
\vspace{-5mm}
\label{fig:acq_compa}
\end{figure*}



\subsubsection{\textcolor{blue}{Evaluation of Zero-Trust Skillset Submission and Verification}}

\textcolor{blue}{For the zero-trust submission and verification framework for AI-model skillsets, the V-Agent achieves a flawless $100\%$ detection accuracy across all $50$ empirical use cases, where all $40$ honest AI models successfully satisfy the cross-verification criteria, while the $10$ adversarial or falsified models are comprehensively intercepted at their respective dishonest dimensions. Let us further quantify the operational overhead across three phases: on-chain skillset submission ($P_1$), off-chain V-Agent verification ($P_2$), and on-chain report anchoring ($P_3$). Consider three submitted models with increasing parameter sizes: denoted as $\#1$ MNIST-LeNet ($177$ KB), $\#2$ EMNIST-CNN ($1.6$ MB), and $\#3$ Vehicle-CNN ($42.7$ MB) for illustration. As a comparative benchmark, the baseline system without CoS (w/o CoS) directly uploads the skillset to decentralized storage layer without any security auditing.}

\textcolor{blue}{As shown in Fig. \ref{fig:sub_compa}, during the on-chain skillset submission phase ($P_1$), CoS incurs a marginal latency overhead over the trust-by-default baseline, bounded within the millisecond range, as the ledger synchronization is solely for the lightweight skillset's cryptographic metadata. Phase $2$ ($P_2$) represents the core security-enhancing primitive unique to CoS, where the off-chain V-Agent performs multidimensional behavioral, structural, and architectural auditing. This pipeline encompasses dataset compatibility testing, model architecture graph validation, empirical performance profiling, and LLM-driven metadata self-consistency verification. Fig. \ref{fig:sub_compa} demonstrates that the execution latency and local CPU utilization time during $P_2$ scale monotonically with the model size, showing the intensified computational complexity of deep forward-pass evaluations on larger weight matrices. In the final phase ($P_3$), the cryptographic verification report is immutably anchored onto the consortium ledger via smart contracts, which demands negligible network traffic and constant-time transaction commitment delays.}

\textcolor{blue}{In conclusion, the performance overhead introduced by TrustAgentNet is heavily dominated by the off-chain verification phase ($P_2$), whereas the blockchain consensus pipeline imposes marginal resource costs. Furthermore, even if the consortium blockchain network scales up in large-scale production environments, potentially increasing endorsement and ordering delays, the total system overhead is projected to remain acceptable, which is fundamentally attributed to the structurally decoupled, asynchronous nature of our on-chain/off-chain integration design.}

\begin{figure*}[htbp] 
\centering
\includegraphics[width=0.8\linewidth]{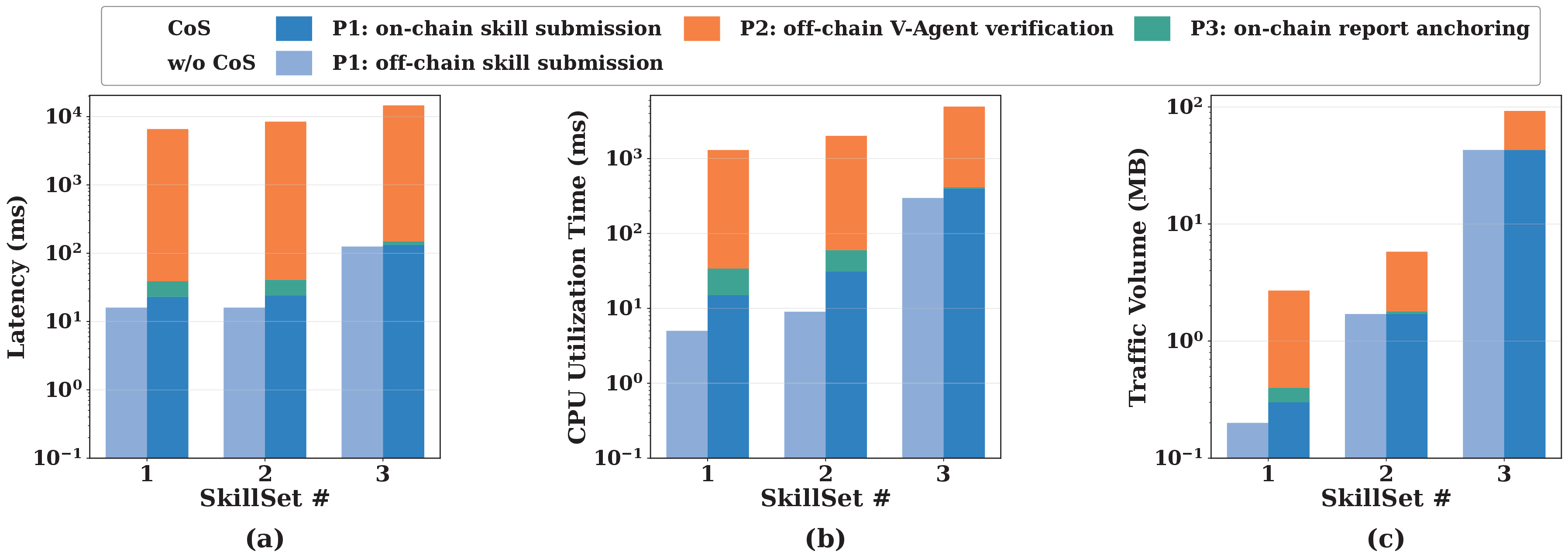}
\caption{\textcolor{blue}{Comparison in (a) latency, (b) CPU utilization time and (c) traffic volume between the proposed CoS and the baseline (w/o CoS) during the skillset submission process across different submitted models. Similarly, the dominant $P_2$ (V-Agent verification) phase is placed at the top of the stacking hierarchy to preserve the visual legibility of the remaining micro-overhead stages.}}
\vspace{-5mm}
\label{fig:sub_compa}
\end{figure*}

\subsubsection{\textcolor{blue}{Empirical Verification Performance on Procedural Skillsets}}
\textcolor{blue}{To demonstrate the architectural generalizability and cross-domain viability of our verification pipeline, we extend our evaluation to non-AI-model procedural capabilities. It is worth noting that while mainstream security literature for agent tool-use often conducts massive-scale fuzzing specifically tailored for procedural ecosystems, this subsection serves as a focused feasibility study. Our primary objective is to verify whether the zero-trust auditing workflow can seamlessly adapt to procedural skillsets without architectural modifications.}

\begin{table}[htbp]
\caption{\textcolor{blue}{Performance Evaluation of the Proposed Verification Protocol on Procedural Skillsets}}
\label{tab:procedural_results}
\centering
\resizebox{\columnwidth}{!}{
\huge
\begin{tabular}{lccccc}
\toprule
\textbf{Feature Category} & \textbf{No. of Features} & \textbf{Accuracy} & \textbf{Precision} & \textbf{Recall} & \textbf{F1-Score} \\ 
\midrule
\textbf{Entry}       & 214  & 96.20\% & 0.995 & 0.969 & 0.982 \\
\textbf{Capability}  & 797  & 84.80\% & 0.971 & 0.880 & 0.923 \\
\textbf{Behavior}    & 257  & 76.80\% & 0.509 & 0.487 & 0.498 \\
\textbf{Constraint}  & 210  & 76.70\% & 0.797 & 0.910 & 0.850 \\ 
\textbf{Overall}     & 1478 & 83.91\% & 0.869 & 0.829 & 0.848 \\ 
\bottomrule
\end{tabular}
}
\end{table} 

\textcolor{blue}{To this end, we evaluate our verification protocol on $171$ procedural skills extracted from the ClawHub repository. To establish a ground-truth baseline, we manually annotate the features encompassing \emph{Entry}, \emph{Capability}, \emph{Behavior}, and \emph{Constraint}, and their cross-layer alignment classifications (i.e., match, miss, and uncertain) for each skill.As shown in Table~\ref{tab:procedural_results}, the verification protocol achieves a robust overall consistency accuracy of $83.91\%$ and an aggregated F1-score of $0.848$ across a total of $1478$ features. Specifically, our protocol delivers high alignment precision for structural features including \emph{Entry} and \emph{Capability}. However, \emph{Behavior} alignment remains the primary bottleneck, which is fundamentally constrained by the ambiguous natural-language behavioral definitions in skill manifests that lack clear programmatic verification criteria. Overall, this empirical insight confirms that while highly effective for static, well-defined features, augmenting the framework with a dynamic sandbox verification remains crucial to enhance runtime behavioral analysis of procedural skills.}

\textcolor{blue}{Fig.~\ref{fig:latency_breakdown} illustrates the latency performance of the verification protocol by analyzing the execution breakdown of 10 representative procedural skillsets across the three phases. The empirical results indicate that the latency associated with the initial artifact decoupling phase is practically negligible (as depicted in the inset diagram). The overall latency is heavily dominated by the LLM-driven inference during the feature extraction and consistency analysis phases. Statistical evaluation indicates that this inference latency scales positively with the feature granularity, yielding a Pearson correlation coefficient $r{=}0.811$ and $R^2{=}0.658$. On average, the protocol consumes approximately $5.4$ seconds of processing time per feature, demonstrating a predictable linear dependency on the functional granularity of the skillset metadata.}

\begin{figure}[htbp]
\centering
\includegraphics[width=1\linewidth]{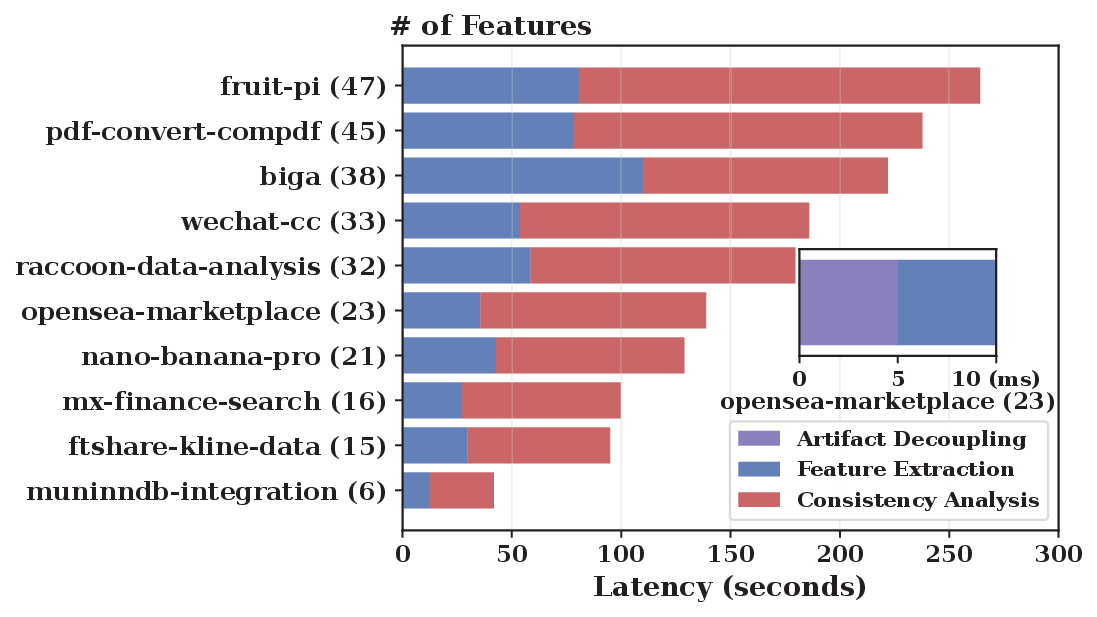}
\caption{\textcolor{blue}{Fine-grained execution latency breakdown for 10 representative procedural skillsets across three operational phases, where the inset highlights the millisecond-level cost of the artifact decoupling phase.}}
\vspace{-5mm}
\label{fig:latency_breakdown}
\end{figure}
\subsubsection{\textcolor{blue}{Case Study: Autonomous Self-Healing Against Multi-Dimensional Attacks}}

Due to the inherent autonomy and sophisticated tool-use capabilities of AI agents, AgentNet faces novel and persistent attack surfaces at the UEs, which can lead to systematic performance degradation or complete failures, ultimately jeopardizing the reliability and integrity of the overall AgentNet systems. We now demonstrate how the proposed TrustAgentNet framework supports autonomous agent recovery when an agent's locally deployed skillsets are compromised by leveraging the CoS and distributed skillset storage. In this implementation, a mission-critical agent is deployed on a UE, possessing a skillset dedicated to high-precision military-civilian vehicle image recognition, used for continuous analysis of surveillance video feeds. \textcolor{blue}{The experiment simulates diverse malicious attacks targeting the locally deployed skillset assets on the UE, including gradient inversion, model poisoning and data poisoning.}

\begin{figure}[htp]
\centering
\includegraphics[width=1\linewidth]{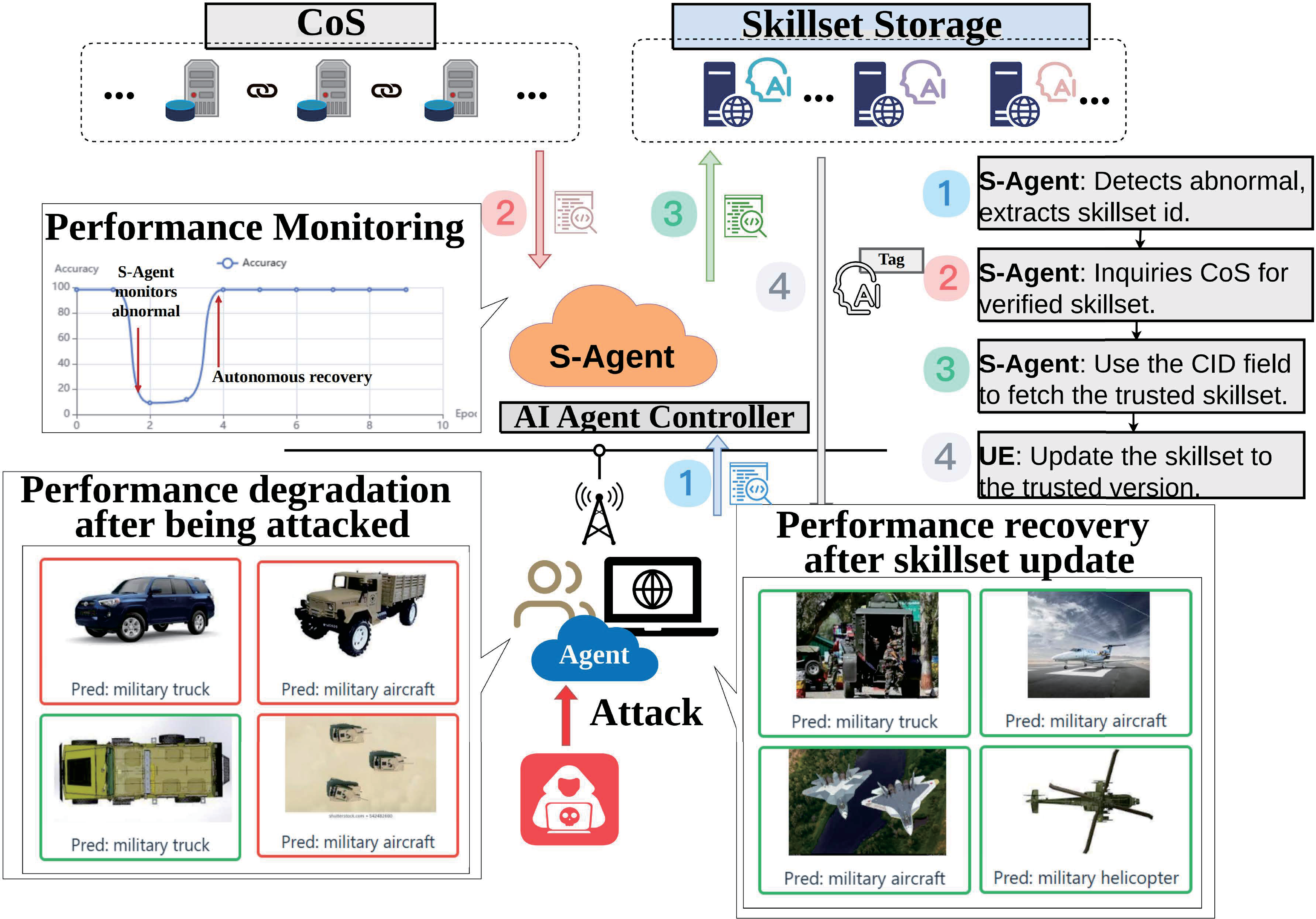}  
\caption{An implementation of TrustAgentNet for autonomous self-healing and recovery after an agent at the UE is attacked: the skillset performance after the attack at the UE is shown in the left-bottom subplot, and the skillset performance after the autonomous recovery is shown in
the right-bottom subplot.}
\label{Figure_recovery}
\vspace{-4mm}
\end{figure}

As illustrated in Fig. \ref{Figure_recovery}, the recovery workflow is orchestrated by a specialized S-Agent, which continuously monitors the operational integrity of the local skillsets. \textcolor{blue}{Crucially, the S-Agent operates in an attack-agnostic manner; rather than relying on specific attack signatures or intrusion detection rules, it solely evaluates the runtime functional KPIs of the active skillsets.} The S-Agent promptly identifies a significant and sudden drop in the skillset’s performance accuracy (i.e., failure to correctly recognize vehicles) \textcolor{blue}{caused by any of the aforementioned attacks}, and initiates an autonomous recovery workflow. Specifically, the  S-Agent extracts the unique \textit{skillset ID} field from the compromised local skillset's metadata. It then submits a secure transaction to the CoS to query the complete, verified skillset tag associated with that ID. The CoS smart contract returns the verified tag, which contains the unique \textit{CID} field of the trusted skillset artifact. The S-Agent utilizes this \textit{CID} to query the decentralized Skillset Storage network, fetching the corresponding original model files, and then securely assembles the recovered model and deploys it to the UE, replacing the compromised local skillset. The autonomous recovery process demonstrates that TrustAgentNet can efficiently and promptly identify agent abnormality and restore full agent functionality following \textcolor{blue}{heterogeneous} malicious attacks. This significantly enhances the security, resilience, and robustness of Agentic AI networking in vulnerable, real-world operational environments.

\subsection{\textcolor{blue}{Empirical Validation of the Three-way Tradeoff for CoC-oriented On-Chain Multi-Agent Collaboration}}
To empirically validate the inherent trade-off between security level, skillset training, and resource consumption within the proposed TrustAgentNet, as analytically established in Section \ref{sec:coc}, we conduct an experiment on our prototype. This experiment involved 20 agents, each equipped with an identical skillset for hand-written digit recognition but possessing a unique non-i.i.d. dataset. For the CoC associated with this skillset, we preconfigured five peer nodes to endorse the participation of agents in the training process. The endorsement policy stipulated that for an agent to be granted \begin{figure}[htp]
\centering
\includegraphics[width=0.9\linewidth]{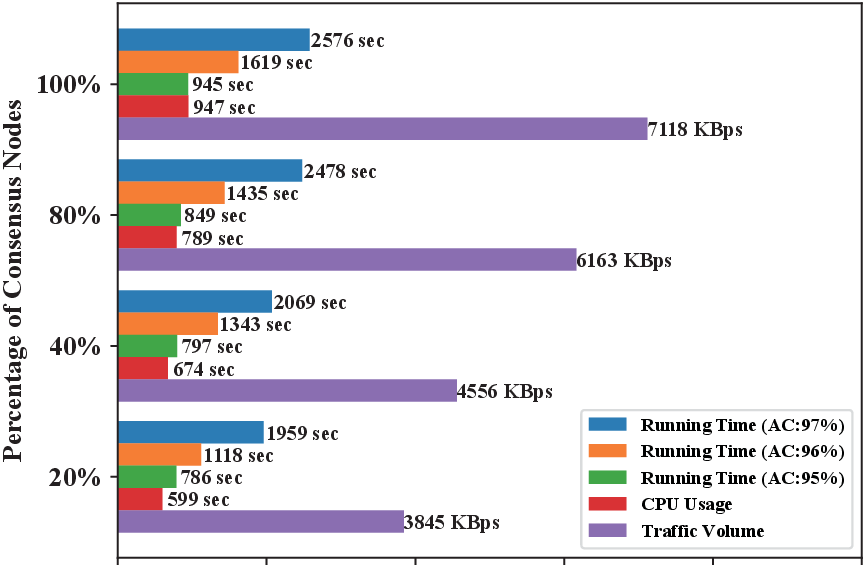}
\caption{Running time, real-measured CPU occupancy time, and traffic volume under different percentage of nodes for agreeing to the endorsement in the TrustAgentNet prototype.}
\label{Figure_consumption_security}
\vspace{-3mm}
\end{figure}access, $x\%$ of these peer nodes must agree. We further assume a linear correlation between the increase in $x\%$ and the number of agents failing to meet the security certification requirements. Each agent that successfully joined the training process performs local model training for 5 epochs before submitting its updated parameters to the CoC.

As depicted in Fig. \ref{Figure_consumption_security}, we conduct a comparative analysis of the real-measured running time, CPU occupancy time, and traffic volume under varying percentages of nodes required to agree on an endorsement. Specifically, running time is defined as the cumulative time duration for the agents to collaboratively train the model until a target accuracy is achieved. CPU occupancy time represents the aggregate time during which the CPUs of the blockchain servers are utilized. Lastly, traffic volume refers to the total volume of data transmitted across the blockchain servers. Fig. \ref{Figure_consumption_security} clearly demonstrates that enhancing system security by demanding agreement from a higher percentage of peer nodes for endorsement leads to increased computational demands (CPU usage) on the blockchain infrastructure. This is a direct consequence of a greater number of nodes needing to validate transactions according to the endorsement policy.
Furthermore, the communication overhead escalates substantially as all endorsing peers are required to broadcast their endorsement decisions and digital signatures across the blockchain network. Notably, the running time for collaboratively training the skillset to achieve target accuracies of $95\%$, $96\%$, and $97\%$ also increases with a greater number of required endorsing peer nodes, indicating that improvements in security come at the expense of the distributed training performance of the skillset.


\section{Conclusion}
\label{sec:conclusion}

This paper proposes TrustAgentNet, a novel consortium-blockchain-based framework to foster zero-trust security in \textcolor{blue}{skillset supply chain and multi-agent collaboration for autonomous AgentNet systems. The framework establishes a hierarchical dual-tier paradigm, leveraging a global CoS for static lifecycle asset validation and transient, task-oriented CoCs for dynamic multi-agent interactions. Lifecycle governance protocols for skillsets on CoS are proposed, specifically leveraging an M-Agent for semantic-to-skillset retrieval and a V-Agent for cross-modal integrity verification. Furthermore, theoretical analysis further reveals the fundamental three-way trade-off among the security level, task performance error, and resource overhead, which is empirically validated. Extensive evaluations on a hardware prototype demonstrate that the zero-trust overhead is heavily dominated by off-chain inference, whereas the optimized blockchain layer imposes minimal computational and consensus costs. Ultimately, empirical results confirm that TrustAgentNet robustly intercepts adversarial attacks and can seamlessly generalize across both AI-model and procedural skillsets, offering a viable trust-by-verification paradigm for AgentNet systems.}



\bibliographystyle{IEEEtran}
\bibliography{Block}
\begin{IEEEbiography}[{\includegraphics[width=1.1in,height=1.3in,clip,keepaspectratio]{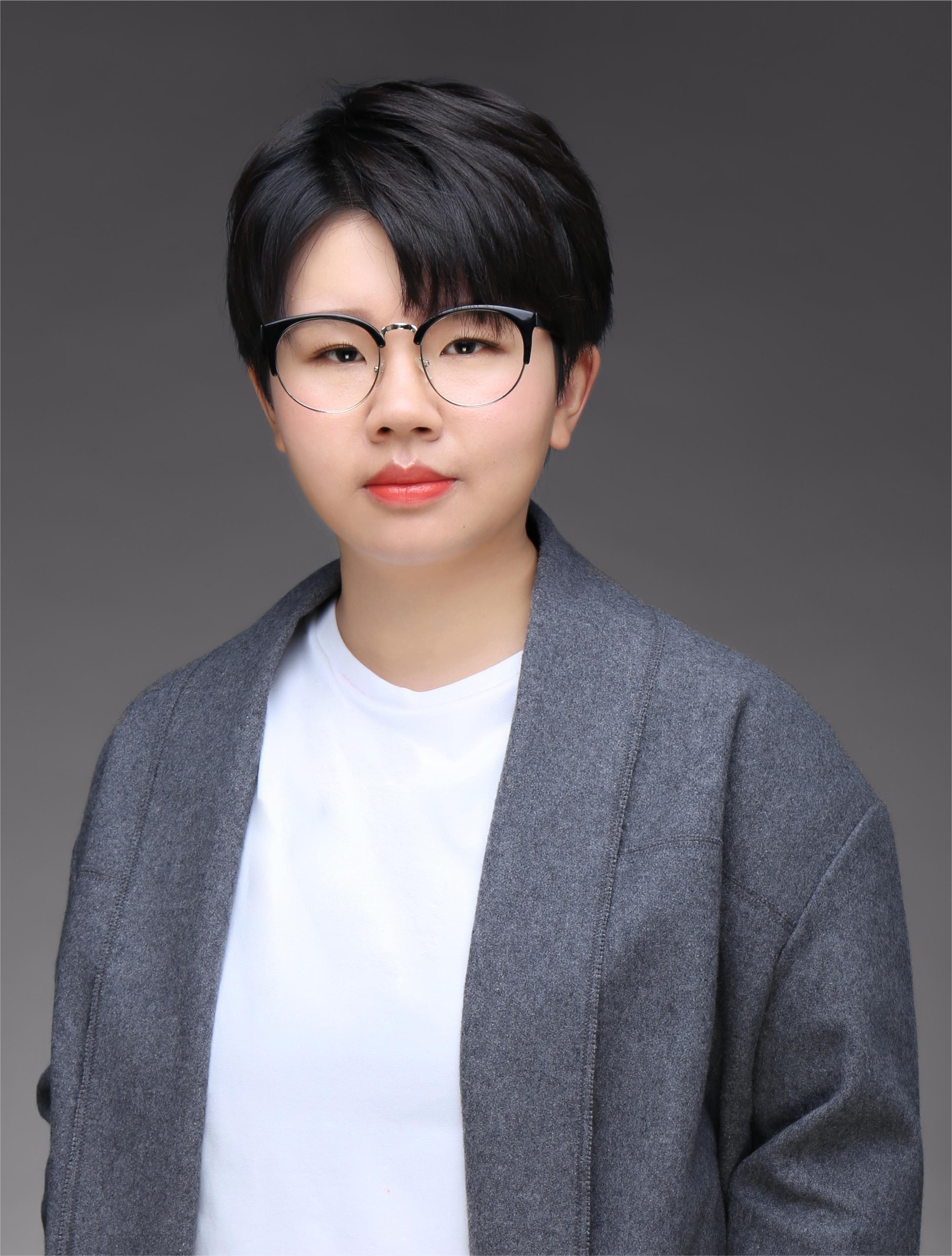}}]{Yayu Gao} (Member, IEEE) is an Associate Professor in the School of Electronic Information and Communications at the Huazhong University of Science and Technology (HUST), Wuhan, China. 
Her research interests include next-generation wireless communication networks, agentic AI communication networking, edge intelligence, and artificial intelligence for networking.
\end{IEEEbiography}
\begin{IEEEbiography}[{\includegraphics[width=1.1in,height=1.3in,clip,keepaspectratio]{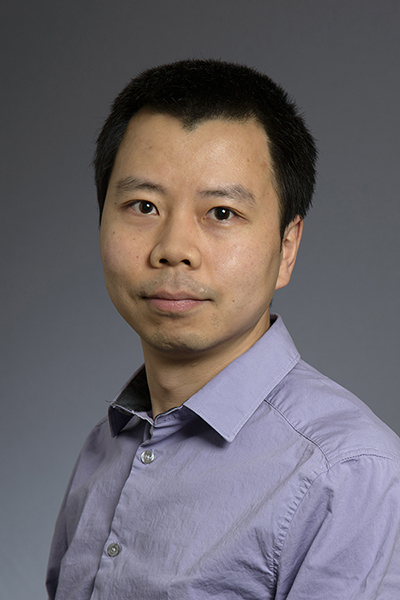}}]{Yong Xiao} (Senior Member, IEEE) 
is a professor in the School of Electronic Information and Communications at the Huazhong University of Science and Technology (HUST), Wuhan, China. He is also with Peng Cheng Laboratory, Shenzhen, China, and Pazhou Laboratory (Huangpu), Guangzhou, China. He is the associate group leader of the network intelligence group of IMT-2030 (6G promoting group). 
His research interests include AI/ML, game theory, distributed optimization, and their applications in agentic AI networking and semantic communications.
\end{IEEEbiography}
\vspace{-1cm}

\begin{IEEEbiography}[{\includegraphics[width=1.1in,height=1.3in,clip,keepaspectratio]{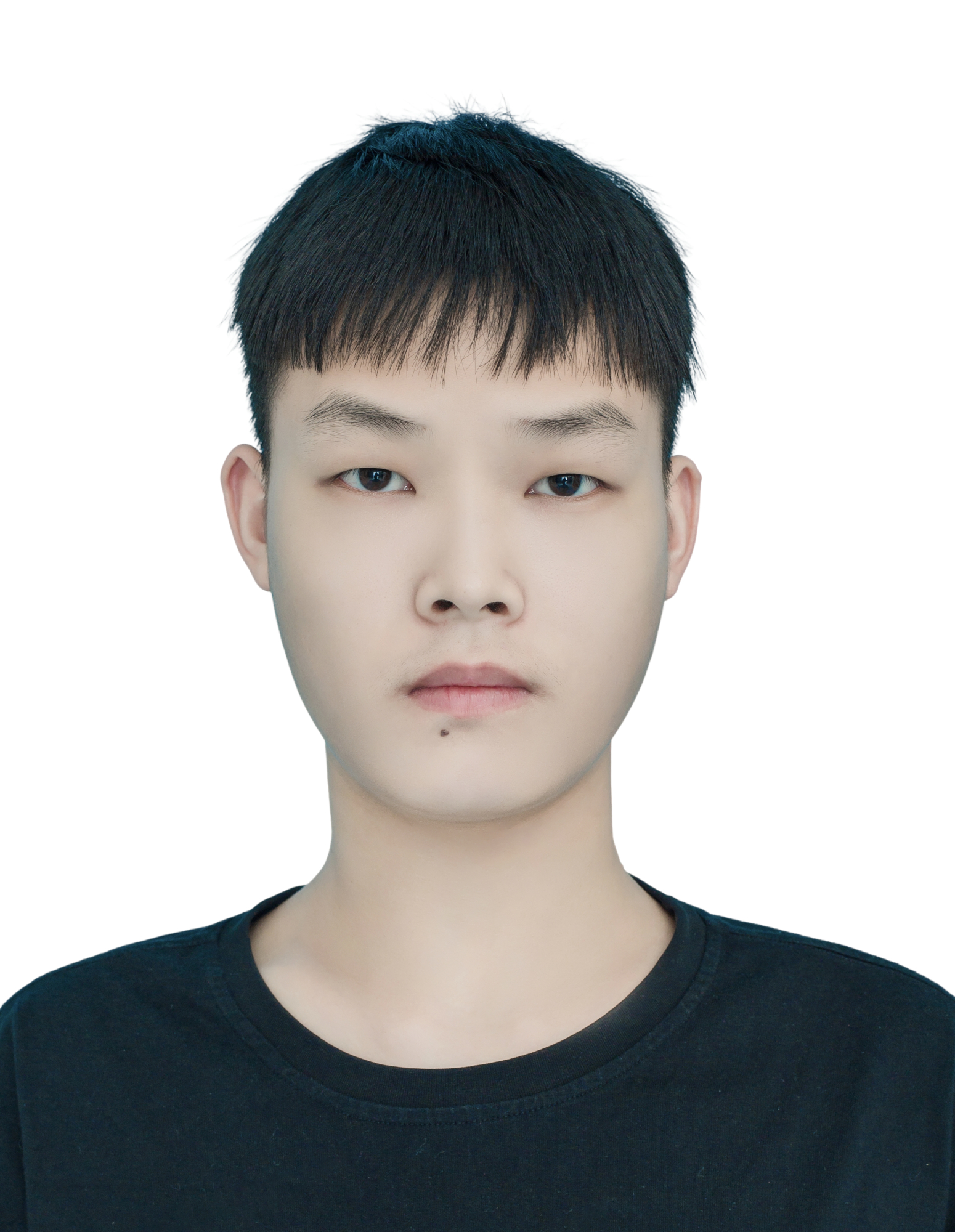}}]{Hao Hu} is currently pursuing the postgraduate degree with the School of Mechanical Engineering and Electronic Information, China University of Geosciences, Wuhan, China. His research interests include trustworthy intelligence, agentic AI networking, and zero-trust systems.
\end{IEEEbiography}
\vspace{-1cm}

\begin{IEEEbiography}[{\includegraphics[width=1.1in,height=1.3in,clip,keepaspectratio]{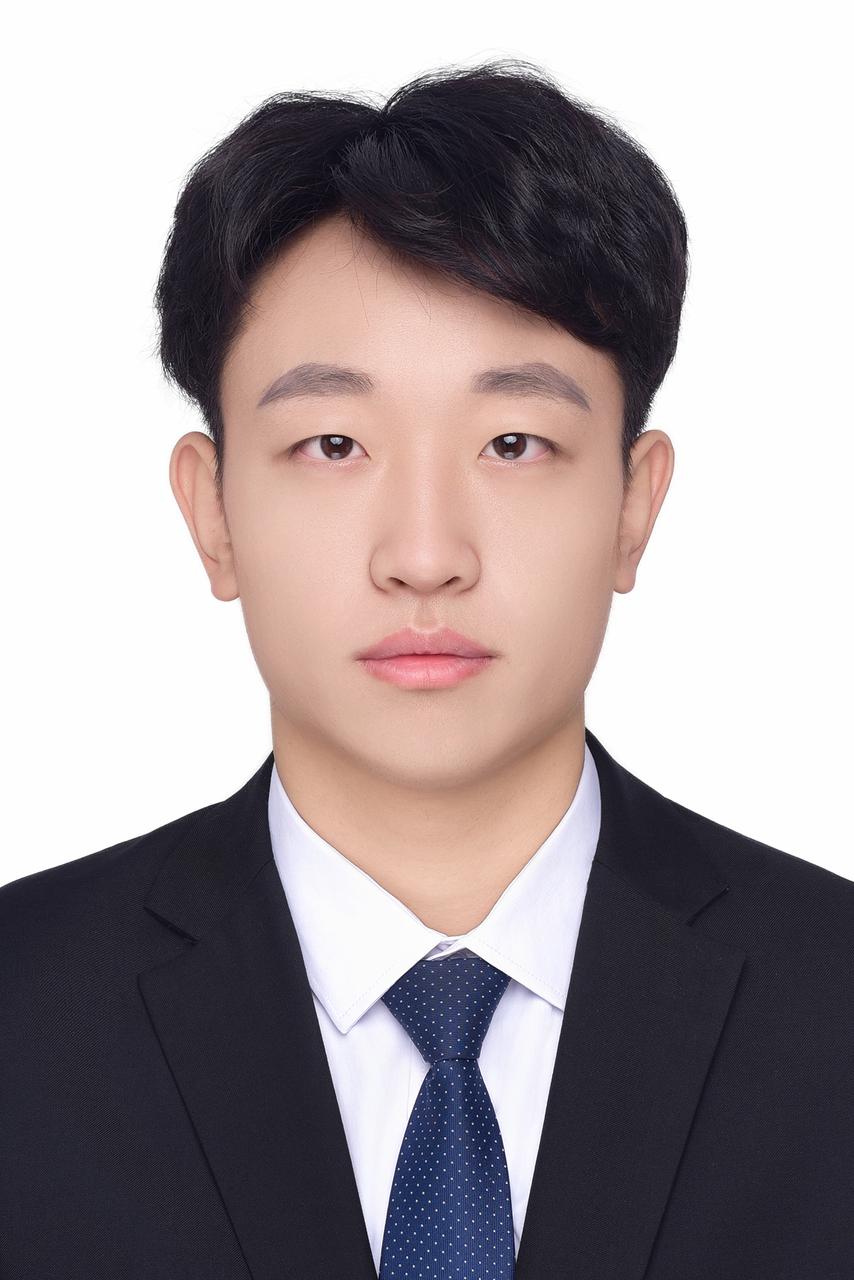}}]{Xubo Li} (Student Member, IEEE) received his B.S. degree in biomedical engineering from Huazhong University of Science and Technology, Wuhan, China in 2023. He is currently pursuing his Ph.D. in the School of Electronic Information and Communications at Huazhong University of Science and Technology, Wuhan, China. His research interests include machine learning, network intelligence, and next-generation wireless communication technology.
\end{IEEEbiography}
\vspace{-1cm}

\begin{IEEEbiography}[{\includegraphics[width=1.1in,height=1.3in,clip,keepaspectratio]{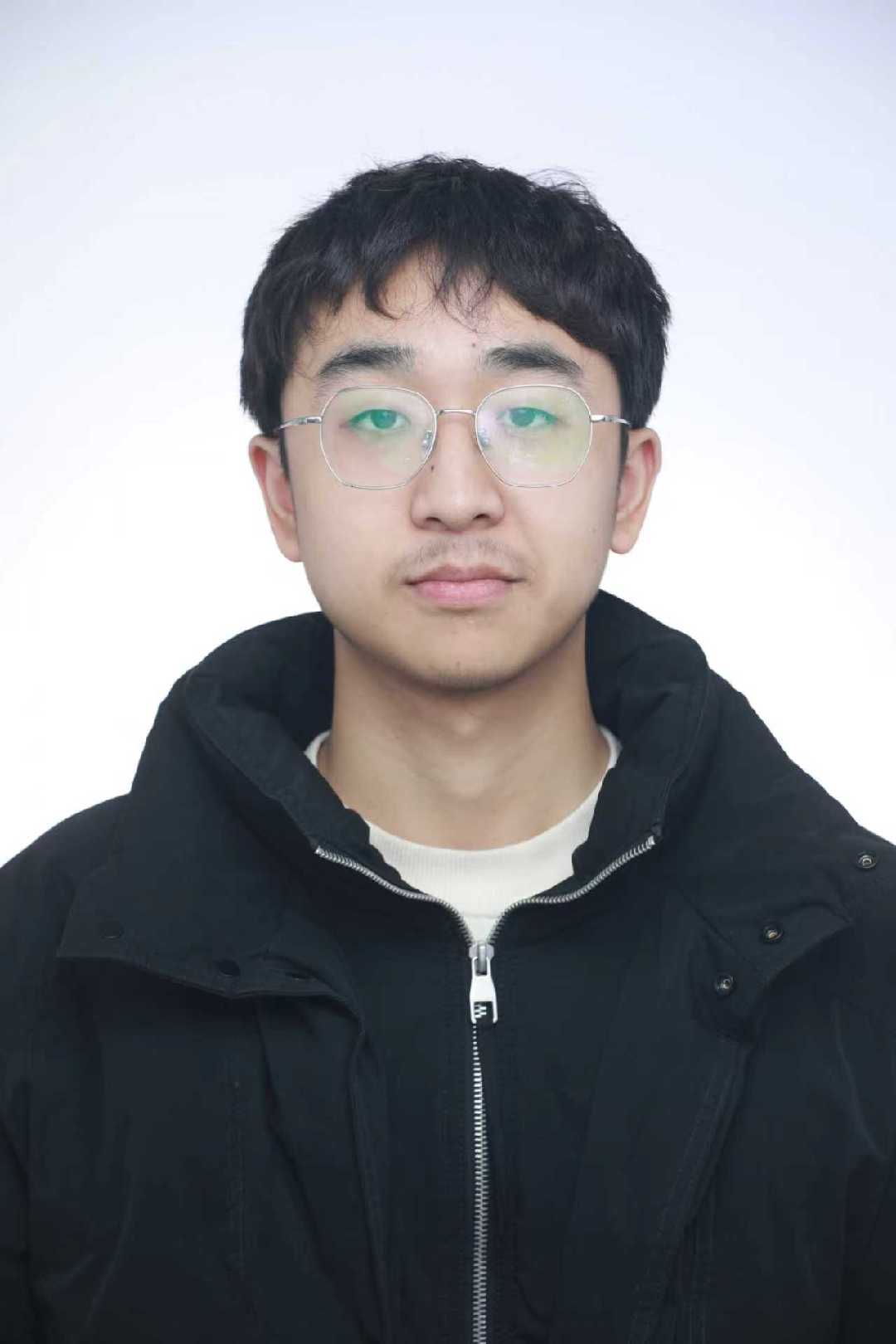}}]{Zhiwei Liu} received the bachelor's degree from Huazhong University of Science and Technology, Wuhan, China. He is currently pursuing a postgraduate degree with the School of Electronic Information and Communication, Huazhong University of Science and Technology, Wuhan. His research interests include trustworthy intelligence, agentic AI networking.
\end{IEEEbiography}

\begin{IEEEbiography}[{\includegraphics[width=1.1in,height=1.3in,clip,keepaspectratio]{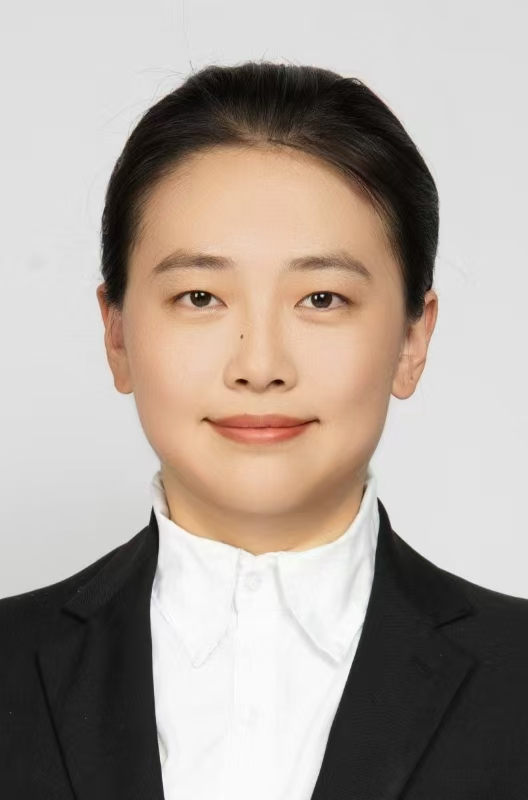}}]{Yingyu Li} (Member, IEEE) 
is an Associate Professor at the School of Mechanical Engineering and Electronic Information, China University of Geosciences (Wuhan). Her research interests include machine learning and artificial intelligence for next-generation wireless networks, federated edge intelligence, green/low-carbon communication networks, distributed optimization, semantic communications, and intelligent Internet of Things.
\end{IEEEbiography}

\begin{IEEEbiography}[{\includegraphics[width=1.1in,height=1.3in, clip,keepaspectratio]{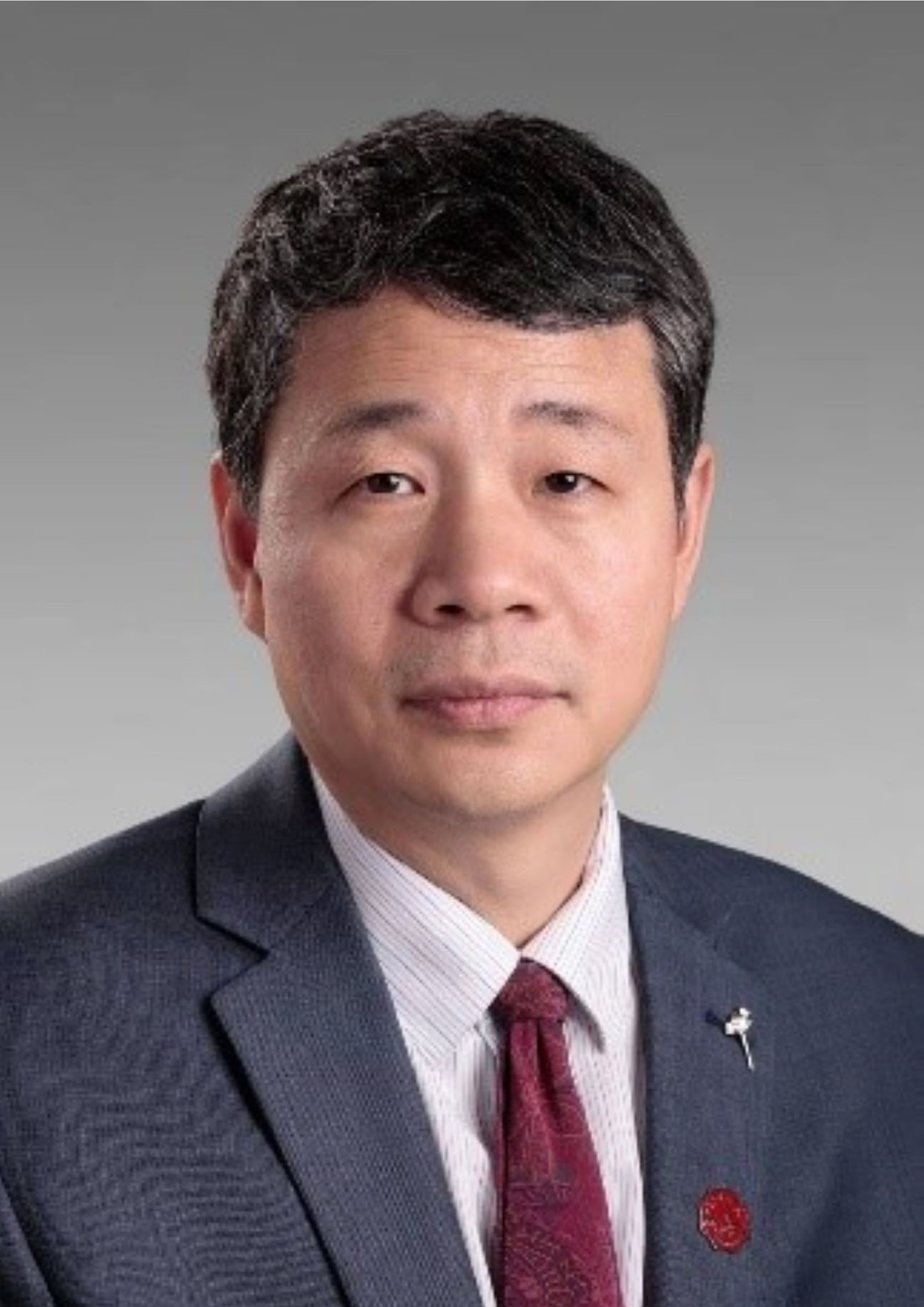}}]{Guangming Shi} (Fellow, IEEE) 
is the Vice Dean of Peng Cheng Laboratory and a Professor with the School of Artificial Intelligence, Xidian University. He is an IEEE Fellow, the chair of IEEE CASS Xi’an Chapter, a senior member of ACM and CCF, Fellow of the Chinese Institute of Electronics, and Fellow of IET. 
His research interests include Artificial Intelligence, Semantic Communications, and Human-Computer Interaction.
\end{IEEEbiography}
\begin{IEEEbiography}[{\includegraphics[width=1.1in,height=1.3in,clip,keepaspectratio]{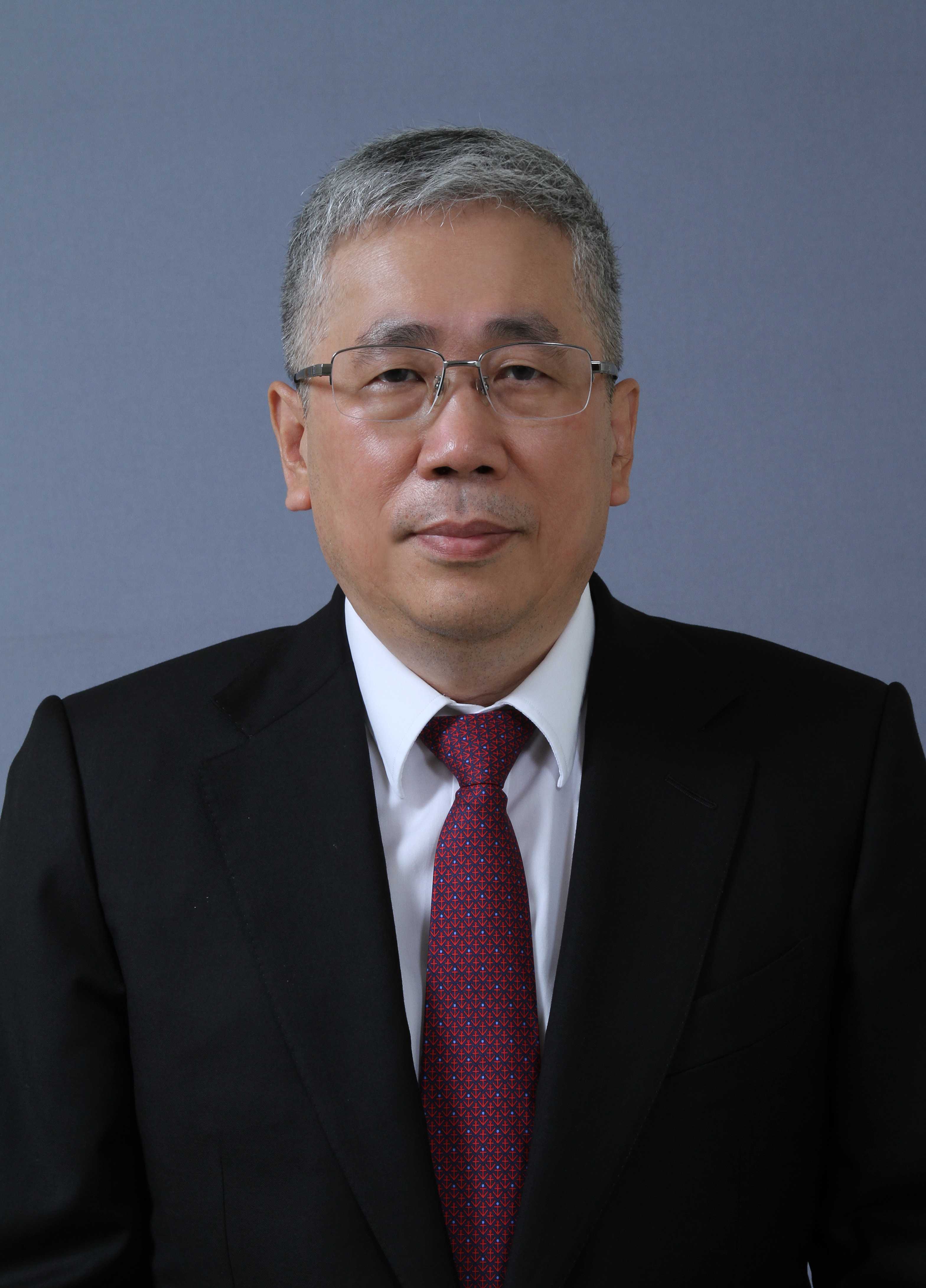}}]{Ping Zhang} (Fellow, IEEE) is a professor in the School of Information and Communication Engineering at the Beijing University of Posts and Telecommunications. He is a member of the National Academy of Engineering of China. He is currently the director of the State Key Laboratory of Networking and Switching Technology, a member of IMT-2020 (5G) Experts Panel, and a member of the Experts Panel for China's 6G development. 
His research is in the broad area of wireless communications with emphasis on novel coding design and model-driven approaches for semantic communications. 
\end{IEEEbiography}

\newpage
\appendices
\section{Proofs of Theorem \ref{thm1}}\label{proof_of_thm1}
Before introducing the performance analysis for the proposed TrustAgentNet under federated global model traing framework, we first present the upper bound for participating agents.
Let $\bar{w}_{\cM_i,t} = \sum_{a_i \in \cM_i} p_{i, k} w_{k, t}$ denote the global model at the $t$th round of coordination.
In \cite{liconvergence}, it has been proved that the following upper bound hold for the $M_i$ agents that participates in the collaboration through authentication. 

\begin{lemma}\label{thm1_lemma1}
Suppose assumptions \ref{smooth}-\ref{variance_bounded} hold, we have
\begin{equation}   
\begin{aligned}
    \mE (\| \bar{w}_{\cM_i,t} - w_{\cM_i}^*\|^2) \le \Delta,
\end{aligned}
\end{equation}
where $\Delta$ is denoted by
\begin{equation}
\begin{aligned}
    &\Delta = \\
    &\frac{4}{\mu^2(\gamma+Et)} \left( \sum_{a_k \in \mathcal{R}_i} \! \frac{p_{i,k}^2 \sigma_k^2}{D_k} \!+\! 8E^2 G^2 \!+\! A_{\mathcal{M}_i} \!+\! 6L C_{\mathcal{M}_i}^*(F_i)\right). \\\nonumber
\end{aligned}
\end{equation}
\end{lemma}
Decomposing the term $\| \bar{w}_{\cM_i,t} - w_{\cM_i}^*\|^2$, we have
\begin{flalign}
    &\mE \| \bar{w}_{\cM_i,t} - w_{\cM_i}^*\|^2 - \mE \| \bar{w}_{\cM_i,t} - w_{\cR_i}^*\|^2 \nonumber\\
    = \ & 2 \langle \bar{w}_{\cM_i,t} - w_{\cM_i}^* + w_{\cM_i}^* - w_{\cR_i}^* \rangle + \| w_{\cR_i}^* - w_{\cM_i}^* \|^2 \nonumber\\
    = \ & 2 \langle \bar{w}_{\cM_i,t} - w_{\cM_i}^* \rangle - \| w_{\cR_i}^* - w_{\cM_i}^* \|^2 \nonumber\\
    \ge \ & -\mE \| \bar{w}_{\cM_i,t} - w_{\cM_i}^*\|^2 - 2 \| w_{\cR_i}^* - w_{\cM_i}^* \|^2, \label{thm1_equation1}
\end{flalign}
where the last inequality comes from the fact that $2ab \ge -a^2 -b^2$.
Combing Lemma \ref{thm1_lemma1} and equation (\ref{thm1_equation1}), we have
\begin{equation}\label{thm1_equation2}
    \mE \| \bar{w}_{\cM_i,t} - w_{\cR_i}^*\|^2 \le 2 \Delta + 2 \| w_{\cR_i}^* - w_{\cM_i}^* \|^2.
\end{equation}
Then, by leveraging the L-smooth property of the loss function, we can obtain
\begin{equation}
\begin{aligned}
    \mathcal{E}_{s_i} 
    &\le \frac{L}{2} \mE \| \bar{w}_{\cM_i,t} - w_{\cR_i}^*\|^2 \\
    & \le L \Delta + L \| w_{\cR_i}^* - w_{\cM_i}^* \|^2.
\end{aligned}
\end{equation}
This concludes the proof.

\section{Proofs of Theorem \ref{thm2}}\label{proof_of_thm2}
\vspace{-1mm}
To establish the performance analysis of our proposed TrustAgentNet under multi-agent model partitioning and sharing framework, we first prove the asymptotic stability of the framework in Lemma \ref{Lemma1}. Building upon this, we then derive the upper bound of the expected gradient norm to characterize the convergence behavior toward a stationary point.
\begin{lemma}\label{Lemma1}
    Suppose Assumption \ref{smooth} holds. Consider the sequence $\{ \bOmega_{i, t} \}_{t=1}^{T}$ generated by static weighting-based solution. 
    It holds that
    \begin{eqnarray}
        &\frac{1}{T} \sum_{t=0}^{T-1} \mathbb{E} [ \| \sum_{a_k \in \cM_i} p_{i,k} \nabla F_{i,k} (\bOmega_{i, t}) \|^2 ] \le\nonumber\\
        & \frac{2}{(2\beta-L\beta^2)T} (C_I + C_{\cM_i}(F_i)). 
    \end{eqnarray}
\end{lemma}

\begin{IEEEproof}
    We first establish that, for any given static weight $p_{i,k}$ for $a_k \in \cM_i$, there exists an upper bound for the gap between the global losses in two consecutive iteration rounds. By leveraging the $L$-smoothness of $F_{i,k} (\cdot)$, we have
    \begin{flalign}
        & \sum_{a_k \in \cM_i} p_{i, k} (F_{i,k} (\bOmega_{i, t+1}) - F_{i,k} (\bOmega_{i, t})) 
        \le \nonumber\\
        & \langle \sum_{a_k \in \cM_i} p_{i, k} \nabla F_{i,k} (\bOmega_{i, t}),  \bOmega_{i, t+1} - \bOmega_{i, t}\rangle 
         + \frac{L}{2}\| \bOmega_{i, t+1} - \bOmega_{i, t} \|^2 \nonumber\\
       & =  - \beta_t \langle \sum_{a_k \in \cM_i} p_{i,k} \nabla F_{i,k} (\bOmega_{i, t}), \sum_{a_k \in \cM_i} p_{i,k} \nabla F_{i,k} (\bOmega_{i, t}) \rangle \nonumber\\
        & + \frac{L}{2} \beta_t^2 \|\sum_{a_k \in \cM_i} p_{i,k} \nabla F_{i,k} (\bOmega_{i, t}) \|^2. \label{lemma1_eq1}
    \end{flalign}

    Taking expectation over the random sample during each iteration on both sides of the above inequality, and setting the step size $\frac{L}{2} \ge \beta_t = \beta \ge 0$, we have
    \begin{flalign}
    & \mathbb{E} [\sum_{a_k \in \cM_i} p_{i, k} (F_{i,k} (\bOmega_{i, t+1}) {-} F_{i,k} (\bOmega_{i, t}))] 
        {\le} \nonumber\\
        &{-} \beta \mathbb{E} [\| \sum_{a_k \in \cM_i}p_{i, k} \nabla F_{i,k} (\bOmega_{i, t}) \|^2] \nonumber\\
        &
         {+} \frac{L_g}{2} \beta^2 \mathbb{E} [\| \sum_{a_k \in \cM_i} p_{i, k} \nabla F_{i,k} (\bOmega_{i, t}) \|^2].
    \end{flalign}
    
    Summing the the above inequality across all coordination rounds $t=0$ to $T-1$ and rearranging terms, we obtain
    \begin{flalign}
        &\frac{1}{T} \sum_{t=0}^{T-1} \mathbb{E}[ \| \sum_{a_k \in \cM_i} p_{i, k} \nabla F_{i,k} (\bOmega_{i, t}) \|^2 ] 
        \le \nonumber\\
        & \frac{2}{2\beta-L\beta^2} \cdot \frac{1}{T}\{ \mathbb{E} [ \sum_{a_k \in \cM_i} p_{i, k} ( F_{i,k} (\bOmega_{i, 0}) - F_{i,k} (\bOmega_{\cR_i, T}) )] \nonumber\\
        &+  \mathbb{E} [ \sum_{a_k \in \cM_i} p_{i, k} ( F_{i,k} (\bOmega_{\cR_i, T} - F_{i,k} (\bOmega_{\cM_i, T}) )]\}\nonumber\\
        & \le \frac{2}{(2\beta-L\beta^2)T} (C_I + C_{\cM_i}(F_i)).
    \end{flalign}
    This concludes the proof.
\end{IEEEproof}

Based on Lemma \ref{Lemma1}, by the Jensen's inequality and the convexity of the square function, as well as the sub-additivity of the square root function, it holds that
\begin{eqnarray}
    &\frac{1}{T} \sum_{t=0}^{T-1} \mathbb{E}[ \| \sum_{a_k \in \cM_i} p_{i, k} \nabla F_{i,k} (\bOmega_{i, t}) \| ]  \nonumber\\
    &\le\left( \frac{1}{T} \sum_{t=0}^{T-1} \mathbb{E}[ \| \sum_{a_k \in \cM_i} p_{i, k} \nabla F_{i,k} (\bOmega_{i, t}) \|^2 ] \right)^{\frac{1}{2}}  \nonumber\\
    & \le \sqrt{\frac{2(C_I + C_{\cM_i}(F_i))}{(2\beta-L\beta^2)T} }.
\end{eqnarray}
This completes the proof.

\end{document}